\def \VersionAuthor {}
\ifdefined\VersionAuthor
	\newcommand{\AuthorVersion}[1]{#1}
	\newcommand{\FinalVersion}[1]{}
\else
	\newcommand{\AuthorVersion}[1]{}
	\newcommand{\FinalVersion}[1]{#1}
\fi

\ifdefined\VersionAuthor
\documentclass[]{article}
\else
\documentclass[a4paper,USenglish,cleveref, autoref, thm-restate]{lipics-v2021}
\fi

\usepackage[utf8]{inputenc}
\usepackage[ruled,vlined,linesnumbered]{algorithm2e}
	\SetKwInOut{Input}{input}
	\SetKwInOut{Output}{output}

\usepackage{subcaption}

\usepackage{paralist} %

\newenvironment{oneenumerate}
	{\ifdefined\VersionLong\begin{enumerate}\else\begin{inparaenum}[1)]\fi}
	{\ifdefined\VersionLong\end{enumerate}\else\end{inparaenum}\fi}

\usepackage{amsmath} %
\usepackage{scalerel}

\usepackage{amssymb} %

\usepackage{mathtools}

\usepackage[misc,geometry]{ifsym} %

\ifdefined \VersionLong
	\newcommand{\LongVersion}[1]{#1}
	\newcommand{\ShortVersion}[1]{}
\else
	\newcommand{\LongVersion}[1]{}
	\newcommand{\ShortVersion}[1]{#1}
\fi

\ifdefined\VersionAuthor
	\usepackage[backend=biber,backref=true,style=alphabetic,url=false,doi=true,defernumbers=true,sorting=anyt,maxnames=99]{biblatex} %
	\makeatletter
	\def\blx@err@patch#1{}
	\makeatother

	\renewbibmacro*{doi+eprint+url}{%
		\iftoggle{bbx:doi}
			{\color{black!40}\footnotesize\printfield{doi}}
			{}%
		\newunit\newblock
		\iftoggle{bbx:eprint}
			{\usebibmacro{eprint}}
			{}%
		\newunit\newblock
		\iftoggle{bbx:url}
			{\usebibmacro{url+urldate}}
			{}%
	}

\fi
\usepackage[svgnames,table]{xcolor}
\definecolor{USPNcobalt}{HTML}{293358}
\definecolor{USPNocre}{HTML}{8b7d6d}
\definecolor{USPNblanc}{HTML}{ffffff}
\definecolor{USPNceruleen}{HTML}{354878}
\definecolor{USPNsable}{HTML}{ad947e}
\ifdefined\VersionAuthor
	\usepackage[
			pdfauthor={Andre, Arcile, Lefaucheux},%
			pdftitle={Execution-time opacity problems in one-clock parametric timed automata},
			breaklinks  = true,
			colorlinks  = true,
			citecolor   = USPNsable,
			linkcolor   = USPNocre,
			urlcolor    = USPNceruleen,
		]{hyperref}
\fi

\ifdefined\VersionAuthor
	\usepackage[capitalise,english,nameinlink]{cleveref} %
	\crefname{line}{\text{line}}{\text{lines}} %
\fi

\newcommand{\defProblem}[3]
{%
\noindent\fcolorbox{black}{USPNsable!15}{
	\begin{minipage}{.95\columnwidth}
		\textbf{#1:}\\
		\textsc{Input}: #2\\
		\textsc{Problem}: #3
	\end{minipage}
}

	\smallskip

}

\usepackage{tikz}
\usetikzlibrary{arrows,automata,positioning}
\tikzstyle{PTA}=[auto, ->, >=stealth']
\tikzstyle{every node}=[initial text=]
\tikzstyle{location}=[rectangle, rounded corners, minimum size=12pt, draw=black, fill=blue!10, inner sep=2pt]
\tikzstyle{invariant}=[draw=black, dotted, inner sep=1pt] %
\tikzstyle{final}=[double, fill=blue!50]
\tikzstyle{symbstate}=[state, draw,rectangle,inner sep=3pt]
\tikzstyle{infinitesymbstate}=[symbstate, fill=blue!10]

\tikzstyle{urgent}=[fill=yellow, thick, dotted] %
\tikzstyle{private}=[fill=red,thick]

\definecolor{coloract}{rgb}{0.50, 0.70, 0.30}
\definecolor{colorclock}{rgb}{0.4, 0.4, 1}
\definecolor{colordisc}{rgb}{1, 0, 1}
\definecolor{colorloc}{rgb}{0.4, 0.4, 0.65}
\definecolor{colorparam}{rgb}{1, 0.6, 0.0}

\newcommand{\styleclock}[1]{\ensuremath{\textcolor{colorclock}{{#1}}}}

\newcommand{\styleparam}[1]{\ensuremath{\textcolor{colorparam}{{#1}}}}

\newcommand{\textstyleact}[1]{\ensuremath{\mathit{#1}}}
\newcommand{\textstyleclock}[1]{\ensuremath{\mathit{#1}}}

\newcommand{\textstyleloc}[1]{\ensuremath{\mathrm{#1}}}
\newcommand{\textstyleparam}[1]{\ensuremath{{#1}}}

\newcommand{\rowHeader}{\rowcolor{blue!20}}

\newcommand{\cellYes}{\cellcolor{green!40}\textbf{$\surd$}}
\newcommand{\cellNo}{\cellcolor{red!40}\textbf{$\times$}}
\newcommand{\cellOpen}{\cellcolor{yellow!40}\textbf{?}}

\ifdefined\VersionAuthor
	\usepackage{thm-restate}
	\usepackage{amsthm}
	\theoremstyle{plain}
	\newtheorem{lemma}{Lemma}
	
	\newtheorem{theorem}{Theorem}
	\newtheorem{corollary}{Corollary}

	\theoremstyle{definition}
	\newtheorem{definition}{Definition}
	\newtheorem{example}{Example}

	\theoremstyle{remark}
	\newtheorem{remark}{Remark}
\fi

\newcommand{\stylealgo}[1]{\ensuremath{\textsf{#1}}}

\newcommand{\EFsynth}{\stylealgo{EFsynth}}

\newcommand{\styleproblem}[1]{\ensuremath{\texttt{#1}}}
\newcommand{\problemTOE}{\styleproblem{$\exists$OE}}
\newcommand{\problemTOS}{\styleproblem{$\exists$OS}}
\newcommand{\problemFTOE}{\styleproblem{FOE}}
\newcommand{\problemWTOE}{\styleproblem{WOE}}
\newcommand{\problemFTOS}{\styleproblem{FOS}}
\newcommand{\problemdTOS}{\styleproblem{d-$\exists$OS}}
\newcommand{\problemdFTOS}{\styleproblem{d-FOS}}

\newcommand{\PET}[1]{\ensuremath{\mathit{PET}(#1)}}
\newcommand{\textPET}[1]{PET}

\newcommand{\assign}{\leftarrow}

\newcommand{\lterm}{\mathit{lt}}
\newcommand{\true}{\ensuremath{\mathit{True}}}
\newcommand{\false}{\ensuremath{\mathit{False}}}
\newcommand{\Rfp}{\mathit{FrP}}

\newcommand{\checkUseMacro}[1]{#1}

\usepackage[english]{nomencl}
\usepackage{acro}

\newcommand{\stylecode}[1]{\textcolor{colorloc}{\texttt{#1}}}

\newcommand{\set}[1]{\ensuremath{\left\{#1\right\}}}

\newcommand{\Time}{\mathbb{T}} %
\newcommand{\setN}{\ensuremath{\mathbb{N}}}

\newcommand{\setR}{\ensuremath{\mathbb{R}}}
\newcommand{\setRgeqzero}{\ensuremath{\setR_{\geq 0}}}

\newcommand{\setZ}{\ensuremath{\mathbb{Z}}}

\newcommand{\compOp}{\bowtie}

\newcommand{\init}{\ensuremath{0}}
\newcommand{\priv}{\ensuremath{{\mathit{priv}}}}
\newcommand{\pub}{\ensuremath{{\mathit{pub}}}}
\newcommand{\target}{\ensuremath{{\mathit{target}}}}
\newcommand{\abs}{\ensuremath{{\mathit{abs}}}}
\newcommand{\final}{\ensuremath{f}}
\newcommand{\valuate}[2]{\ensuremath{#2(#1)}}

\newcommand{\styleAutomaton}[1]{\ensuremath{\mathcal{#1}}}

\newcommand{\clock}{\ensuremath{\textstyleclock{x}}}

\newcommand{\clocki}[1]{\ensuremath{\textstyleclock{\clock_{#1}}}}
\newcommand{\ClockCard}{H} %
\newcommand{\clockval}{\ensuremath{\mu}}
\newcommand{\ClockSet}{\ensuremath{\mathbb{X}}} %
\newcommand{\ClocksZero}{\ensuremath{\vec{0}}}

\newcommand{\clockabs}{\ensuremath{\textstyleclock{\clock_{\abs}}}}

\newcommand{\resets}{\ensuremath{R}}
\newcommand{\reset}[2]{\ensuremath{[#1]_{#2}}}

\newcommand{\numConstant}{\ensuremath{\gamma}}

\newcommand{\param}{\ensuremath{\textstyleparam{p}}}
\newcommand{\parami}[1]{\ensuremath{\textstyleparam{\param_{#1}}}} %

\newcommand{\ParamCard}{\ensuremath{M}} %
\newcommand{\pval}{\ensuremath{v}}

\newcommand{\ParamSet}{\ensuremath{\mathbb{P}}} %

\newcommand{\paramabs}{\ensuremath{\textstyleparam{d}}}

\newcommand{\action}{\ensuremath{\textstyleact{a}}}
\newcommand{\ActionSet}{\ensuremath{\Sigma}}

\newcommand{\constraint}{\ensuremath{C}}

\newcommand{\Diff}{\ensuremath{\mathit{Diff}}}

\newcommand{\LpSl}{LpSl}

\newcommand{\edge}{\ensuremath{\checkUseMacro{e}}}
\newcommand{\edgei}[1]{\ensuremath{\checkUseMacro{\edge_{#1}}}}
\newcommand{\EdgeSet}{\ensuremath{E}}

\newcommand{\guard}{\ensuremath{g}}

\newcommand{\invariant}{\ensuremath{I}}

\newcommand{\loc}{\ensuremath{\textstyleloc{\ell}}}
\newcommand{\loci}[1]{\ensuremath{\textstyleloc{\loc_{#1}}}}
\newcommand{\locinit}{\ensuremath{\textstyleloc{\loc_\init}}}
\newcommand{\locfinal}{\ensuremath{\textstyleloc{\loc_\final}}}
\newcommand{\locpriv}{\ensuremath{\textstyleloc{\loc_\priv}}}
\newcommand{\locpub}{\ensuremath{\textstyleloc{\loc_\pub}}}

\newcommand{\LocSet}{\ensuremath{L}}
\newcommand{\LocsTarget}{\ensuremath{\LocSet_\target}} %

\newcommand{\longuefleche}[1]{\stackrel{#1}{\longrightarrow}}

\newcommand{\varrun}{\checkUseMacro{\rho}} %

\newcommand{\duration}{\ensuremath{\mathit{dur}}} %
\newcommand{\runduration}[1]{\ensuremath{\checkUseMacro{\duration}(#1)}}

\newcommand{\PTA}{\ensuremath{\checkUseMacro{\styleAutomaton{A}}}}
\newcommand{\PTAprivextend}{\ensuremath{\left(\ActionSet, \LocSet, \locinit,\locpriv, \locfinal, \ClockSet, \ParamSet, \invariant, \EdgeSet\right)}}

\newcommand{\styleSymbStatesSet}[1]{\ensuremath{\mathbf{#1}}}
\newcommand{\Constr}{\ensuremath{\styleSymbStatesSet{C}}}

\newcommand{\symbstate}{\ensuremath{\checkUseMacro{\styleSymbStatesSet{s}}}}

\newcommand{\SymbStateSet}{\ensuremath{\styleSymbStatesSet{S}}}
\newcommand{\symbstateinit}{\ensuremath{\checkUseMacro{\symbstate_\init}}}

\newcommand{\SymbTransitions}{\ensuremath{\Rightarrow}} %

\newcommand{\projectP}[1]{\ensuremath{#1{\downarrow_{\ParamSet}}}}
\newcommand{\timelapse}[1]{\ensuremath{#1^\nearrow}}
\newcommand{\PZG}[1]{\ensuremath{\styleSymbStatesSet{PZG}(#1)}}

\newcommand{\K}{K}

\newcommand{\semantics}[1]{\ensuremath{\mathfrak{T}_{#1}}}
\newcommand{\semanticsextend}{\ensuremath{\left(\StateSet, \concstateinit, \ActionSet \cup \setRgeqzero, \transition\right)}}

\newcommand{\transition}{{\ensuremath{\rightarrow}}}
\newcommand{\transitionWith}[1]{\stackrel{#1}{\mapsto}}

\newcommand{\StateSet}{\ensuremath{\mathfrak{S}}}
\newcommand{\concstate}{\ensuremath{\mathfrak{s}}}

\newcommand{\concstateinit}{\ensuremath{\concstate_\init}}

\newcommand{\paramd}{\textstyleparam{\ensuremath{d}}}

\newcommand{\PrivVisit}[1]{\ensuremath{\mathit{Visit}^{\mathit{priv}}(#1)}}
\newcommand{\PubVisit}[1]{\ensuremath{\mathit{Visit}^{\overline{\mathit{priv}}}(#1)}}

\newcommand{\PrivDurVisit}[1]{\ensuremath{\mathit{DVisit}^\mathit{priv}(#1)}}
\newcommand{\PubDurVisit}[1]{\ensuremath{D\mathit{Visit}^{\overline{\mathit{priv}}}(#1)}}

\newcommand{\execTimeText}{\checkUseMacro{execution time}}
\newcommand{\execTimesText}{\checkUseMacro{execution times}}
\newcommand{\execTimes}{\ensuremath{D}}

\newcommand{\arbitrarilyMany}{\ensuremath{*}}

\newcommand{\textof}[1]{\ac{#1}} 					%
\newcommand{\Textof}[1]{\Ac{#1}} 					%
\newcommand{\fullTextOf}[1]{\acf*{#1}} 			%
\DeclareAcronym{pta}{
	short=PTA,
	long=parametric timed automaton,
	short-plural=s,
	long-plural-form=parametric timed automata,
	cite=AHV93,
	extra={\vref{def:PTA}},
	tag={models}
}

\newcommand{\PTAtext}{\textof{pta}}

\DeclareAcronym{ta}{
	short=TA,
	long=timed automaton,
	short-plural=s,
	long-plural-form=timed automata,
	cite=AD94,
	extra={\vref{def:TA}},
	tag={models}
}

\newcommand{\TAtext}{\textof{ta}}

\DeclareAcronym{ppta}{
	short=(P)TA,
	long=(possibly parametric) \acs*{ta},
	short-plural=s,
	long-plural-form=(possibly parametric) \acsp*{ta},
}

\DeclareAcronym{lts}{
	short=LTS,
	long=labeled transition system,
	short-plural=s,
	long-plural-form=labeled transition systems,
	short-indefinite={an},
	long-indefinite={a},
	cite=Keller76,
	extra={\vref{def:LTS}},
	tag={models}
}

\DeclareAcronym{dfa}{
	short=DFA,
	long=deterministic finite-state automaton,
	short-plural=s,
	long-plural-form=deterministric finite-state automata,
	extra={\vref{def:DFA}},
	tag={models}
}

\DeclareAcronym{tts}{
	short=TTS,
	long=timed transition system,
	short-plural=s,
	long-plural-form=timed transition systems,
	cite=HMP91,
	extra={\vref{def:TTS}},
	tag={models}
}
\newcommand{\TTStext}{\textof{tts}}

\DeclareAcronym{pzg}{
	short=PZG,
	long=parametric zone graph,
	extra={\vref{def:PTA:symbolic}},
	tag={misc}
}

\DeclareAcronym{opacity}{
	short=ET-opacity,
	long=execution-time opacity,
	tag={notion},
	extra={\vref{def:opacity:TOSEM:ET-opacity}},
	post={\acuse{opaque}}
}
\DeclareAcronym{opaque}{
	short=ET-opaque,
	long=execution-time opaque,
	tag={notion}, %
	post={\acuse{opacity}}
}

\newcommand{\opaqueText}{\textof{opaque}}
\newcommand{\opacityText}{\textof{opacity}}

\newcommand{\opaqueTextdef}{\fullTextOf{opaque}}

\newcommand{\OpacityTextForSec}{\Textof{opacity}}

\newcommand{\existentialOpaqueText}{\checkUseMacro{\ensuremath{\exists}-\acs*{opaque}}}
\newcommand{\existentialOpacityText}{\checkUseMacro{\ensuremath{\exists}-\acs*{opacity}}}

\newcommand{\fullOpaqueText}{\checkUseMacro{fully \acs*{opaque}}}
\newcommand{\fullOpacityText}{\checkUseMacro{full \acs*{opacity}}}

\newcommand{\FullOpacityText}{\checkUseMacro{Full \acs*{opacity}}}

\newcommand{\existentialOpacityTextForSec}{\checkUseMacro{\ensuremath{\exists}-\acs*{opacity}}}

\newcommand{\fullOpacityTextForSec}{\checkUseMacro{full \acs*{opacity}}}

\DeclareAcronym{tempopacity}{
	short=exp-\acs*{opacity},
	long=expiring \acl*{opacity},
	tag={notion},
	extra={\vref{def:opacity:ICECCS:temporary-timed-opacity}},
	post={\acuse{tempopaque}}
}
\DeclareAcronym{tempopaque}{
	short=exp-\acs*{opaque},
	long=expiring \acl*{opaque},
	tag={notion}, %
	post={\acuse{tempopacity}}
}

\newcommand{\synthesisProblem}[2]{#1 #2 synthesis problem} %
\newcommand{\existentialOpacityParamSynthesisProblem}{\checkUseMacro{\synthesisProblem{\existentialOpacityText{}}{p}}} %
\newcommand{\FullOpacityParamSynthesisProblem}{\checkUseMacro{\synthesisProblem{\FullOpacityText{}}{p}}}

\newcommand{\existentialOpacitySynthesisProblem}{\checkUseMacro{\synthesisProblem{\existentialOpacityText{}}{p-d}}} %
\newcommand{\FullOpacitySynthesisProblem}{\checkUseMacro{\synthesisProblem{\FullOpacityText{}}{p-d}}}

\newcommand{\emptinessProblemOneParam}[1]{#1 emptiness problem}
\newcommand{\emptinessProblem}[2]{\emptinessProblemOneParam{#1 #2}}
\newcommand{\existentialOpacityParamEmptinessProblem}{\checkUseMacro{\emptinessProblem{\existentialOpacityText{}}{p}}}

\newcommand{\FullOpacityParamEmptinessProblem}{\checkUseMacro{\emptinessProblem{\FullOpacityText{}}{p}}}

\newcommand{\styleComplexity}[1]{{\sffamily\upshape #1}}
\newcommand{\EXPSPACE}{\styleComplexity{EXPSPACE}}
 \newcommand{\threeNEXPTIME}{\styleComplexity{3NEXPTIME}}
\newcommand{\NEXPTIME}{\styleComplexity{NEXPTIME}}

\ifdefined\AuthorVersion
	\usepackage{orcidlink}
	\newcommand{\orcidID}[1]{\orcidlink{#1}}

	\usepackage{fontawesome}
	\newcommand{\homepage}[1]{\href{#1}{\color{gray}\faHome}}
\fi

\ifdefined \VersionWithComments
 	\definecolor{colorok}{RGB}{80,80,150}
\else
	\definecolor{colorok}{RGB}{0,0,0}
\fi

\newcommand{\eg}{\textcolor{colorok}{e.g.,}\xspace}
\newcommand{\etal}{\textcolor{colorok}{\emph{et al.}}\xspace}
\newcommand{\ie}{\textcolor{colorok}{i.e.,}\xspace}
\newcommand{\suchthat}{\textcolor{colorok}{s.t.}\xspace}

\newcommand{\wrt}{\textcolor{colorok}{w.r.t.}\xspace}

\ifdefined\VersionAuthor
	\title{Execution-time opacity problems in one-clock parametric timed automata\footnote{%
		This is the author (and extended) version of the manuscript of the same name published in the proceedings of the 44th IARCS Annual Conference on Foundations of Software Technology and Theoretical Computer Science (FSTTCS 2024).
		This work is partially supported by the ANR-NRF French-Singaporean research program ProMiS (ANR-19-CE25-0015 / 2019 ANR NRF 0092) and by ANR BisoUS (ANR-22-CE48-0012).
	}%
}
\else
	\title{Execution-time opacity problems in one-clock parametric timed automata}

\author{\'Etienne Andr\'e}{Universit\'e Sorbonne Paris Nord, LIPN, CNRS UMR 7030, F-93430 Villetaneuse, France\\Institut Universitaire de France (IUF) \and \url{https://lipn.univ-paris13.fr/~andre/}}{}{https://orcid.org/0000-0001-8473-9555}{}%

\author{Johan Arcile}{IBISC, Univ Evry, Universit\'e Paris-Saclay, 91025 Evry, France}{johan.arcile@univ-evry.fr}{https://orcid.org/0000-0001-9979-3829}{}

\author{Engel Lefaucheux%
	}{Universit\'e de Lorraine, CNRS, Inria, LORIA, F-54000 Nancy, France\and\url{https://elefauch.github.io/}}{}{https://orcid.org/0000-0003-0875-300X}{} %

\authorrunning{\'E.\ Andr\'e, J.\ Arcile and E.\ Lefaucheux} %

\Copyright{\'Etienne Andr\'e, Johan Arcile and Engel Lefaucheux} %
\funding{This work is partially supported by the ANR-NRF French-Singaporean research program ProMiS (ANR-19-CE25-0015 / 2019 ANR NRF 0092) and by ANR BisoUS (ANR-22-CE48-0012).}



\category{} %

\relatedversion{} %
\EventEditors{John Q. Open and Joan R. Access}
\EventNoEds{2}
\EventLongTitle{42nd Conference on Very Important Topics (CVIT 2016)}
\EventShortTitle{CVIT 2016}
\EventAcronym{CVIT}
\EventYear{2016}
\EventDate{December 24--27, 2016}
\EventLocation{Little Whinging, United Kingdom}
\EventLogo{}
\SeriesVolume{42}
\ArticleNo{23}
\fi

\begin{document}

\ifdefined\VersionAuthor
	\author{}
	\date{}
\fi

\sloppy

\maketitle

\ifdefined\VersionAuthor
	\noindent{}\textbf{Étienne André\orcidID{0000-0001-8473-9555}}
	\\
	{\em\small{}Université Sorbonne Paris Nord, LIPN, CNRS UMR 7030, Villetaneuse, France}  %
	\\
	{\em\small{}Institut Universitaire de France (IUF)}

	\smallskip

	\noindent{}\textbf{Johan Arcile\orcidID{0000-0001-9979-3829}}
	\\
	{\em\small{}IBISC, Univ Evry, Universit\'e Paris-Saclay, 91025 Evry, France}

	\smallskip

	\noindent{}\textbf{Engel Lefaucheux\orcidID{0000-0003-0875-300X}}
	\\
	{\em\small{}Universit\'e de Lorraine, CNRS, Inria, LORIA, F-54000 Nancy, France}
\fi

\begin{abstract}
	Parametric timed automata (PTAs) extend the concept of timed automata, by allowing timing delays not only specified by concrete values but also by parameters, allowing the analysis of systems with uncertainty regarding timing behaviors.
	The full execution-time opacity is defined as the problem in which an attacker must never be able to deduce whether some private location was visited, by only observing the execution time.
	The problem of full ET-opacity emptiness (\ie{} the emptiness over the parameter valuations for which full execution-time opacity is satisfied) is known to be undecidable for general PTAs.
	We therefore focus here on one-clock PTAs with integer-valued parameters over dense time.
	We show that the full ET-opacity emptiness is undecidable for a sufficiently large number of parameters, but is decidable for a single parameter, and exact synthesis can be effectively achieved.
	Our proofs rely on a novel construction as well as on variants of Presburger arithmetics.
	We finally prove an additional decidability result on an existential variant of execution-time opacity.
\end{abstract}

\ifdefined\VersionWithComments
\fi

\ifdefined\VersionAuthor
\else
	\newpage
\fi

\section{Introduction}\label{section:introduction}

\LongVersion{%
	Numerous critical information systems rely on communication via a shared network\LongVersion{, such as the Internet}.
	Data passing through such networks is often sensitive, and requires secrecy.
}%
\LongVersion{%
	If not handled carefully, information such as private data, authentication code, timing information or localization can be accessible to anyone on the network.
}%
As surveyed in~\cite{BGN17}, for some systems, private information may be deduced simply by observation of public information.
For example, it may be possible to infer the content of some memory space from the access times of a cryptographic module.

The notion of \emph{opacity} \cite{Mazare04,BKMR08} concerns information leaks from a system to an attacker; that is, it expresses the power of the attacker to deduce some secret information based on some publicly observable behaviors.
If an attacker observing a subset of the actions cannot deduce whether a given sequence of actions has been performed, then the system is opaque.
Time particularly influences the deductive capabilities of the attacker.
It has been shown in~\cite{GMR07} that it is possible for models that are opaque when timing constraints are omitted, to be non-opaque when those constraints are added to the models.

For this reason, the notion is extended to \emph{timed} opacity in~\cite{Cassez09}, where the attacker can also observe time\LongVersion{, hence untimed system that are opaque may not be so once time is added}.
The input model is timed automata (TAs)~\cite{AD94}, a formalism extending finite-state automata with real-time variables called \emph{clocks}.
It is proved in~\cite{Cassez09} that this version of timed opacity is undecidable for~TAs.

In~\cite{ALMS22}, a less powerful version of opacity is proposed, where the attacker has access only to the system execution time and aims at deducing whether a private location was visited during the system execution.
This version of timed opacity is called \emph{execution-time opacity (ET-opacity)}.
Two main problems are considered in~\cite{ALMS22}:
\begin{oneenumerate}
	\item the existence of at least one execution time for which the system is ET-opaque (\emph{$\exists$-ET-opacity}), and
	\item whether \emph{all} execution times are such that the system is ET-opaque (called \emph{full ET-opacity}).
\end{oneenumerate}%
\LongVersion{%
	Additionally, in~\cite{ALLMS23}, \emph{weak} ET-opacity is defined as a variant where the attacker may deduce that the private location was \emph{not} visited, but of course not that the private location was visited.
}%
These \LongVersion{three}\ShortVersion{two} notions of opacity are proved to be decidable for TAs~\cite{ALLMS23}.
\LongVersion{%

}%
In the same works, the authors then extend ET-opacity to parametric timed automata (PTAs)~\cite{AHV93}.
PTAs are an extension of TAs where timed constraints can be expressed with timing parameters instead of integer constants, allowing to model uncertainty or lack of knowledge\LongVersion{ about such constants}.
The \LongVersion{three}\ShortVersion{two} problems come with two flavors:
\begin{oneenumerate}%
	\item \emph{emptiness} problems: whether the set of parameter valuations guaranteeing a given version of opacity ($\exists$-ET-opacity or full ET-opacity) is empty or not, and
	\item \emph{synthesis} problems: synthesize all parameter valuations for which a given version of opacity holds.
\end{oneenumerate}%
\ShortVersion{Both emptiness problems \problemTOE{} ($\exists$-ET-opacity emptiness) and \problemFTOE{} (full-ET-opacity emptiness)}\LongVersion{All three problems \problemTOE{} ($\exists$-ET-opacity emptiness), \problemFTOE{} (full-ET-opacity emptiness) and \problemWTOE{} (weak-ET-opacity emptiness)}
have been shown to be undecidable for PTAs, while decidable subclasses are exhibited~\cite{ALMS22,ALLMS23}.
A semi-algorithm (\ie{} that may not terminate, but is correct if it does) is provided to solve $\exists$-ET-opacity synthesis (hereafter \problemTOS{}) in \cite{ALMS22}.

\subsection{Contributions}
We address here full-ET-opacity emptiness (\problemFTOE{}) and synthesis (\problemFTOS{}), and $\exists$-ET-opacity emptiness (\problemTOE{}) and synthesis (\problemTOS{}), for PTAs with integer-valued parameters over dense time with the following main theoretical contributions:
\begin{enumerate}
	\item We prove that \textbf{\problemFTOE{} is undecidable} (\cref{corollary-FTOE-undecidable}) for PTAs with a single clock and a sufficiently large number of parameters\LongVersion{, even by restricting the model to integer parameter valuations}.
	\item We prove in contrast that \textbf{\problemFTOE{} is decidable} (\cref{corollary-FTOE-decidable}) for PTAs with a single clock and a single \LongVersion{integer-valued }parameter.
	\item We prove that \textbf{\problemTOE{} is decidable} (\cref{theorem:TOE-div}) for PTAs with a single clock and arbitrarily many \LongVersion{integer-valued }parameters.
		We also exhibit a better complexity for a single parameter over discrete time (\cref{theorem:TOE-div:GH21}).
\end{enumerate}

We focus on one-clock PTAs, as virtually all problems are undecidable for 3~clocks~\cite{Andre19STTT}, and the 2-clock case is an extremely difficult problem, already for reachability~\cite{GH21}.
Our contributions are summarized in \cref{table:summary-problems}.
In order to prove these results, we improve on the semi-algorithm from~\cite{ALMS22} for \problemTOS{} and provide one for \problemFTOS{}.
These solutions are based on the novel notion of \emph{parametric execution times} (\textPET\ ).
The \textPET{} of a PTA is the total elapsed time and associated parameter valuations on all paths between two given locations.
We provide a semi-algorithm for the computation of \textPET{}, 
that builds upon reachability synthesis (\ie the synthesis of parameter valuations for which a set of locations are reachable) for which a semi-algorithm already exists (\cite{JLR15}).
We then show how to resolve \problemTOS{} and \problemFTOS{} problems by performing set operations on \textPET{} of two complementary subsets of the PTA where we respectively consider only private paths and only non-private paths.

We then solve the full ET-opacity emptiness (\problemFTOE{}) problem for PTAs with 1~clock and 1 \LongVersion{integer-valued }parameter, by rewriting the problems in a parametric variant of Presburger arithmetic.
This is done by
\begin{oneenumerate}%
	\item providing a sound and complete method for encoding infinite \textPET{} for PTAs with 1 clock and arbitrarily many parameters over dense time; and
	\item translating them into parametric semi-linear sets, a formalism defined and studied in~\cite{L23}.
\end{oneenumerate}%
With these ingredients, we notably prove that:
\begin{oneenumerate}
	\item \problemFTOE{} is undecidable in general for PTAs with 1 clock and sufficiently many parameters\LongVersion{, even integer-valued}.
This is done by reducing a known undecidable problem of parametric Presburger arithmetic
(whose undecidability comes from Hilbert's 10th problem) to the \problemFTOE{} problem in this context.
	\item \problemTOE{} is decidable for PTAs with 1 clock and arbitrarily many \LongVersion{integer-valued }parameters.
This is done by reducing \problemTOE{} to the existential fragment of Presburger arithmetic with divisibility, \LongVersion{which is }known to be decidable.
\end{oneenumerate}

\subsection{Related works}\label{ss:related}

\LongVersion{%
\subsubsection{Timed opacity in timed automata}
}

The undecidability of timed opacity proved in~\cite{Cassez09} leaves hope for decidability only by modifying the problem (as in~\cite{ALMS22,ALLMS23}), or by restraining the model.
\LongVersion{%
	In~\cite{WZ17}, Wang \etal\ define a notion of initial-state opacity that is proved decidable on real-time automata~\cite{Dima01} (RTAs).
	RTAs are a strict subclass of TA with a single clock reset at each transition\LongVersion{, meaning the only timed information represented by this formalism is the time elapsed in the current location; most properties are decidable, even complement and language inclusion (undecidable for the full class of TAs)}.
	In this context, an intruder can observe the elapsed time on a given subset of transitions, and the system is opaque if an intruder cannot determine whether the system starts or not from a given secret state.
	In~\cite{WZA18} the authors extend this work to timed language-opacity à la~\cite{Cassez09}\ShortVersion{ and prove decidability of the problem}.
	\LongVersion{%
		The positive decidability properties of RTAs allow to check the emptiness of the inclusion between the language accepted by the system 	and the language where secret words are not allowed.
		Therefore, timed language-opacity is decidable on the RTA subclass.
	}

}\ShortVersion{%
	In~\cite{WZ17,WZA18}, (initial state) opacity is shown to be decidable on a restricted subclass of TAs called real-time automata~\cite{Dima01}.
}%
\LongVersion{%
	In~\cite{AEYM21}, an intruder have access to a subset of actions, along with timed information, as in~\cite{Cassez09}.
	The originality of~\cite{AEYM21} is to consider a \emph{time-bounded} framework.
	\LongVersion{
		As in~\cite{Cassez09}, a secret location is timed opaque if the intruder cannot infer from the observation of any execution that the system has reached this particular location.
	}%
	Two variants of this definition of opacity are proposed.
	The first one, called \emph{timed bounded opacity} implies that the location is opaque up to a given time duration, and is decidable for non-Zeno TA\LongVersion{ (\ie{} TA where it is impossible for an infinity of actions to occur in finite time)}.
	The second one, \emph{$\delta$-duration bounded opacity} implies that the location must remain secret for at least $\delta$ time units after it is reached.
	This problem can then be solved by timed bounded inclusion checking\LongVersion{, which is decidable for timed automata~\cite{ORW09};
	this is contrast with the (potentially unbounded) language inclusion checking, known to be undecidable~\cite{AD94}}.
	This time-bounded setting is the crux that explains the difference of decidability between~\cite{Cassez09} and~\cite{AEYM21}.
}\ShortVersion{%
	In~\cite{AEYM21}, a notion of \emph{timed bounded opacity}, where the secret has an expiration date, and over a time-bounded framework, is proved decidable.
}%
Opacity over subclasses of TAs (such as one-clock or one-actions TAs) is considered in~\cite{ADL24,AGWZH24} and over discrete time in~\cite{KKG24}.

\LongVersion{%
\subsubsection{Timed opacity in parametric timed automata}
}
In \cite{ALMS22}, $\exists$-ET-opacity synthesis ($\problemTOS$) is solved using a semi-algorithm.
The method is based on a self-composition of the PTA with $m$ parameters and $n$ clocks, where the resulting model consists of $m+1$ 
parameters and $2n+1$ clocks.
The method terminates if the symbolic state space of this self-composition is finite.
Our work proposes in contrast an approach based on set operations on parametric execution times (\textPET{}) of both complementary subsets of the PTA where we respectively consider only private paths and only non-private paths.
Those submodels are each composed of $m+1$ parameters and $n+1$ clocks.
Our new method terminates if the symbolic state spaces of both submodels are finite.
Another improvement is that the method described here also supports full timed opacity synthesis ($\problemFTOS$).%

\LongVersion{%
\subsubsection{Decision problems for parametric timed automata}
}
The reachability emptiness problem (\ie{} the emptiness over the valuations set for which a given target location is reachable) is known to be undecidable in general since~\cite{AHV93}.
The rare decidable settings require a look at the number of parametric clocks (\ie{} compared at least once in a guard or invariant to a parameter), non-parametric clocks and parameters; throughout this paper, we denote these 3~numbers using a triple $(pc, npc, p)$.
Reachability emptiness is decidable
	for $(1, \arbitrarilyMany, \arbitrarilyMany)$-PTAs (``$\arbitrarilyMany$'' denotes ``arbitrarily many'' for decidable cases, and ``sufficiently many'' for undecidable cases) over discrete time~\cite{AHV93} or dense time with integer-valued parameters~\cite{BBLS15},
	for $(1, 0, \arbitrarilyMany)$-PTAs over dense time over rational-valued parameters~\cite{ALM20},
		and
	for $(2, \arbitrarilyMany, 1)$-PTAs over discrete time~\cite{BO17,GH21};
and it is undecidable
	for $(3, \arbitrarilyMany, 1)$-PTAs over discrete or dense time~\cite{BBLS15},
		and
	for $(1, 3, 1)$-PTAs over dense time only for rational-valued parameters~\cite{Miller00}.
See~\cite{Andre19STTT} for a complete survey as of~2019.

\LongVersion{%
\subsection{Outline}
}\ShortVersion{\medskip}
\cref{section:preliminaries} recalls the necessary preliminaries\LongVersion{, \ie{} syntax and semantics of PTAs, as well as definitions of execution-time opacity that where proposed in~\cite{ALMS22,ALLMS23}}.
\cref{section:PETS} introduces one of our main technical proof ingredients, \ie{} the definition of \textPET{}, and \textPET{}-based semi-algorithms for \problemTOS{} and \problemFTOS{}.
\cref{section:TOS} considers the \problemFTOE{} problem over $(1, 0, \arbitrarilyMany)$-PTAs (undecidable) and $(1, 0, 1)$-PTAs (decidable).
\cref{section:TOE} proves decidability of \problemTOE{} for $(1, 0, \arbitrarilyMany)$-PTAs.
We also give a better complexity for $(1, 0, 1)$-PTAs over discrete time.
\cref{section:conclusion} concludes.

\section{Preliminaries}\label{section:preliminaries}

\LongVersion{%
\subsection{Clocks, parameters and guards}\label{ss:clocks}
}

We let $\Time$ be the domain of the time, which will be either non-negative reals $\setRgeqzero$ (continuous-time semantics) or naturals $\setN$ (discrete-time semantics).
Unless otherwise specified, we assume $\Time = \setRgeqzero$.

\emph{Clocks} are real-valued variables that all evolve over time at the same rate.
We assume a set~$\ClockSet = \{ \clocki{1}, \dots, \clocki{\ClockCard} \} $ of \emph{clocks}.
A \emph{clock valuation} is a function
$\clockval : \ClockSet \rightarrow \Time$\LongVersion{, assigning a non-negative value to each clock}.
We write $\ClocksZero$ for the clock valuation assigning $0$ to all clocks.
Given a constant $\numConstant \in \Time$, $\clockval + \numConstant$ denotes the valuation \suchthat\ $(\clockval + \numConstant)(\clock) = \clockval(\clock) + \numConstant$, for all $\clock \in \ClockSet$.
Given $\resets \subseteq \ClockSet$, we define the \emph{reset} of a valuation~$\clockval$, denoted by $\reset{\clockval}{\resets}$, as follows: $\reset{\clockval}{\resets}(\clock) = 0$ if $\clock \in \resets$, and $\reset{\clockval}{\resets}(\clock)=\clockval(\clock)$ otherwise.

A \emph{(timing) parameter} is an unknown %
integer-valued constant of a model.
We assume a set~$\ParamSet = \{ \parami{1}, \dots, \parami{\ParamCard} \} $ of \emph{parameters}.
A \emph{parameter valuation} $\pval$ is a function %
$\pval : \ParamSet \rightarrow \setN$.

We assume ${\compOp} \in \{<, \leq, =, \geq, >\}$.
A \emph{clock guard}~$\constraint$ is a conjunction of inequalities over $\ClockSet \cup \ParamSet$ of the form
$\clock \compOp \sum_{1 \leq i \leq \ParamCard} \alpha_i \parami{i} + \numConstant$, with
$\clock \in  \ClockSet$, 
$\parami{i} \in \ParamSet$,
and
$\alpha_i, \numConstant \in \setZ$.
Given~$\constraint$, we write~$\clockval\models\pval(\constraint)$ if %
the expression obtained by replacing each~$\clock$ with~$\clockval(\clock)$ and each~$\param$ with~$\pval(\param)$ in~$\constraint$ evaluates to true.

\subsection{Parametric timed automata}

Parametric timed automata (PTAs) extend TAs with parameters within guards and invariants in place of integer constants~\cite{AHV93}.
We also add to the standard definition of PTAs a special private location, which will be used to define our subsequent opacity concepts.

\begin{definition}[PTA~\cite{AHV93}]
	A \PTAtext{} $\PTA$ is a tuple \mbox{$\PTA = \PTAprivextend$}, where:
	\begin{oneenumerate}
		\item $\ActionSet$ is a finite set of actions;
		\item $\LocSet$ is a finite set of locations;
		\item $\locinit \in \LocSet$ is the initial location;
		\item $\locpriv \in \LocSet$ is a special private location;
		\item $\locfinal \in \LocSet$ is the final location;
		\item $\ClockSet$ is a finite set of clocks;
		\item $\ParamSet$ is a finite set of parameters;
		\item $\invariant$ is the invariant, assigning to every $\loc\in \LocSet$ a clock guard $\invariant(\loc)$ (called \emph{invariant});
		\item $\EdgeSet$ is a finite set of edges  $\edge = (\loc,\guard,\action,\resets,\loc')$
		where~$\loc,\loc'\in \LocSet$ are the source and target locations, $\action \in \ActionSet$,
		$\resets\subseteq \ClockSet$ is a set of clocks to be reset, and $\guard$ is a clock guard.
	\end{oneenumerate}
	\label{def:PTA}
\end{definition}
\begin{figure}[tb]
	\centering

	\begin{subfigure}[b]{0.45\textwidth}
	\centering
	\begin{tikzpicture}[PTA, scale=1.5, xscale=1, yscale=.4]

		\node[location, initial] at (0, -1) (s0) {$\loci{0}$};

		\node[location, private] at (1, 0) (s2) {$\locpriv$};

		\node[location, final] at (2, -1) (s1) {$\locfinal$};

		\node[invariant, below=of s0,yshift=2.8em] {$\styleclock{\clock} \leq 3$};
		\node[invariant, above=of s2,yshift=-2.8em] {$\styleclock{\clock} \leq \styleparam{\parami{2}}$};

		\path (s0) edge[bend left] node[align=center]{$\styleclock{\clock} \geq \styleparam{\parami{1}}$} (s2); %
		\path (s0) edge[] node[align=center]{} (s1); %
		\path (s2) edge[bend left] node[align=center]{} (s1); %

	\end{tikzpicture}

	\caption{A PTA example~$\PTA$}
	\label{figure:example-PTA}
	\end{subfigure}
	\hfill{}
	\begin{subfigure}[b]{0.45\textwidth}
	\centering
	\begin{tikzpicture}[PTA, scale=1.5, xscale=1, yscale=.4]

		\node[location, initial] at (0, -1) (s0) {$\loci{0}$};

		\node[location, private] at (1, 0) (s2) {$\locpriv$};

		\node[location, final, urgent] at (2, -1) (s1) {$\locfinal$};

		\node[invariant, below=of s0,yshift=2.8em] {$\styleclock{\clock} \leq 3$};
		\node[invariant, above=of s2,yshift=-2.8em] {$\styleclock{\clock} \leq \styleparam{\parami{2}}$};

		\path (s0) edge[bend left] node[align=center]{$\styleclock{\clock} \geq \styleparam{\parami{1}}$} (s2); %
		\path (s0) edge[] node[below,align=center]{$\clockabs = \paramabs$} (s1); %
		\path (s2) edge[bend left] node[align=center]{$\clockabs = \paramabs$} (s1); %

	\end{tikzpicture}

	\caption{$\PTA'$}
	\label{figure:example-PTA-transformed-PET}
	\end{subfigure}
	\hfill{}
	\begin{subfigure}[b]{0.45\textwidth}
	\centering
	\begin{tikzpicture}[PTA, scale=1.5, xscale=1, yscale=.4]

		\node[location, initial] at (0, -1) (s0) {$\loci{0}$};

		\node[location, private] at (1, 0) (s2) {$\locpriv$};

		\node[location, final] at (2, -1) (s1) {$\locfinal$};

		\node[invariant, below=of s0,yshift=2.8em] {$\styleclock{\clock} \leq 3$};
		\node[invariant, above=of s2,yshift=-2.8em] {$\styleclock{\clock} \leq \styleparam{\parami{2}}$};

		\path (s0) edge[bend left] node[align=center]{$\styleclock{\clock} \geq \styleparam{\parami{1}}$\\$b \assign \true$} (s2); %
		\path (s0) edge[] node[below,align=center]{$b = \true$} (s1); %
		\path (s2) edge[bend left] node[align=center]{$b = \true$} (s1); %

	\end{tikzpicture}

	\caption{$\PTA^{\locpriv}_{\locfinal}$}
	\label{figure:example-PTA-transformed-priv}
	\end{subfigure}
	\hfill{}
	\begin{subfigure}[b]{0.45\textwidth}
	\centering
	\begin{tikzpicture}[PTA, scale=1.5, xscale=1, yscale=.4]

		\node[location, initial] at (0, -1) (s0) {$\loci{0}$};

		\node[location, final] at (2, -1) (s1) {$\locfinal$};

		\node[invariant, below=of s0,yshift=2.8em] {$\styleclock{\clock} \leq 3$};

		\path (s0) edge[] node[align=center]{} (s1); %

	\end{tikzpicture}

	\caption{$\PTA^{\neg \locpriv}_{\locfinal}$}
	\label{figure:example-PTA-transformed-pub}
	\end{subfigure}

	\caption{A PTA example and its transformed versions. The yellow dotted location is urgent.}
\end{figure}
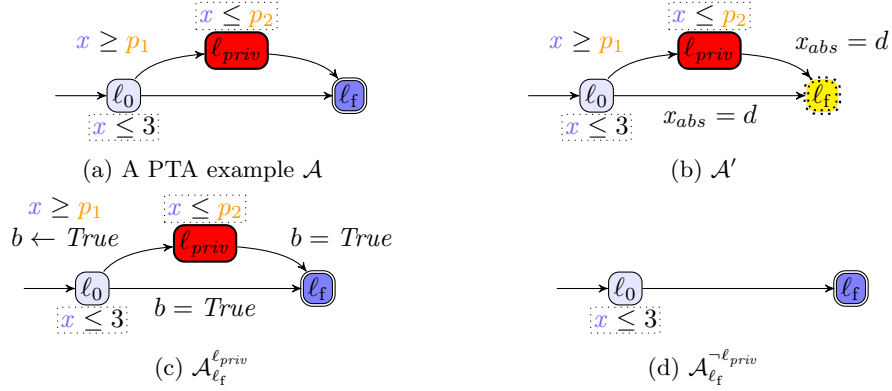

Given a parameter valuation~$\pval$, we denote by $\valuate{\PTA}{\pval}$ the non-parametric structure where all occurrences of a parameter~$\parami{i}$ have been replaced by~$\pval(\parami{i})$.
\begin{definition}[Reset-free PTA]\label{def:reset-freePTA}
	A \emph{reset-free PTA} $\PTA = \PTAprivextend$ is a PTA where $\forall\ (\loc,\guard,\action,\resets,\loc') \in \EdgeSet$, $\resets = \emptyset$.
\end{definition}
\begin{example}
	Consider the PTA~$\PTA$ in \cref{figure:example-PTA}.
	It has three locations, one clock and two parameters (actions are omitted).
	``$\clock \leq \parami{2}$'' is the invariant of~$\locpriv$, and the transition from~$\loc_0$ to~$\locpriv$ has guard ``$\clock \geq \parami{1}$''.
	In this example, $\clock$ is never reset, and therefore $\PTA$ happens to be reset-free.
\end{example}

\LongVersion{%
Let us now recall the concrete semantics of PTAs.
}

\begin{definition}[Semantics of a \TAtext{}]
	Given a \PTAtext{} $\PTA = \PTAprivextend$ and a parameter valuation~$\pval$,
	the semantics of the TA $\valuate{\PTA}{\pval}$ is given by the \TTStext{} $\semantics{\valuate{\PTA}{\pval}} = \semanticsextend$, with
	\begin{enumerate}
		\item $\StateSet = \{ (\loc, \clockval) \in \LocSet \times \setRgeqzero^\ClockCard \mid \clockval \models \valuate{\invariant(\loc)}{\pval} \}$,
		\LongVersion{\item }$\concstateinit = (\locinit, \ClocksZero) $,
		\item  $\transition$ consists of the discrete and (continuous) delay transition relations:
		\begin{enumerate}
			\item discrete transitions: $(\loc,\clockval) \transitionWith{\edge} (\loc',\clockval')$,
			if $(\loc, \clockval) , (\loc',\clockval') \in \StateSet$, and there exists $\edge = (\loc,\guard,\action,\resets,\loc') \in \EdgeSet$, such that $\clockval'= \reset{\clockval}{\resets}$, and $\clockval\models\pval(\guard$).
			\item delay transitions: $(\loc,\clockval) \transitionWith{\numConstant} (\loc, \clockval+\numConstant)$, with $\numConstant \in \setRgeqzero$, if $\forall \numConstant' \in [0, \numConstant], (\loc, \clockval+\numConstant') \in \StateSet$.
		\end{enumerate}
	\end{enumerate}
\end{definition}

Moreover we write $(\loc, \clockval)\longuefleche{(\numConstant, \edge)} (\loc',\clockval')$ for a combination of a delay and discrete transition if
$\exists  \clockval'' :  (\loc,\clockval) \transitionWith{\numConstant} (\loc,\clockval'') \transitionWith{\edge} (\loc',\clockval')$.

Given a TA~$\valuate{\PTA}{\pval}$ with concrete semantics $\semanticsextend$, we refer to the states of~$\StateSet$ as the \emph{concrete states} of~$\valuate{\PTA}{\pval}$.
A \emph{run} of~$\valuate{\PTA}{\pval}$ is an alternating sequence of concrete states of~$\valuate{\PTA}{\pval}$ and pairs of edges and delays starting from the initial state $\concstateinit$ of the form
$(\loci{0}, \clockval_{0}), (d_0, \edge_0), (\loci{1}, \clockval_{1}), \cdots$
with
$i = 0, 1, \dots$, $\edge_i \in \EdgeSet$, $d_i \in \setRgeqzero$ and
$(\loci{i}, \clockval_{i}) \longuefleche{(d_i, \edge_i)} (\loci{i+1}, \clockval_{i+1})$.

Given a state~$\concstate = (\loc, \clockval)$, we say that $\concstate$ is \emph{reachable} in~$\valuate{\PTA}{\pval}$ if $\concstate$ appears in a run of~$\valuate{\PTA}{\pval}$.
By extension, we say that $\loc$ is reachable in~$\valuate{\PTA}{\pval}$; and by extension again, given a set~$\LocsTarget$ of locations, we say that $\LocsTarget$ is reachable in~$\valuate{\PTA}{\pval}$ if there exists $\loc \in \LocsTarget$ such that $\loc$ is reachable in~$\valuate{\PTA}{\pval}$.

Given a finite run $\varrun : (\loci{0}, \clockval_{0}), (d_0, \edgei{0}), (\loci{1}, \clockval_{1}), \cdots, (d_{i-1}, \edgei{i-1}), (\loci{n}, \clockval_{n})$, the \emph{duration} of $\varrun$ is $\runduration{\varrun} = \sum_{0 \leq i \leq n-1} d_i$.
We also say that $\loci{n}$ is reachable in time~$\runduration{\varrun}$.

\LongVersion{%
\subsection{Symbolic semantics}\label{ss:symbolic}
}

Let us now recall the symbolic semantics of PTAs (see \eg{} \cite{HRSV02}).
\LongVersion{%
\subsubsection{Constraints}

}%
We first\LongVersion{ need to} define operations on constraints.
A \emph{linear term} over $\ClockSet \cup \ParamSet$ is of the form $\sum_{1 \leq i \leq \ClockCard} \alpha_i \clock_i + \sum_{1 \leq j \leq \ParamCard} \beta_j \param_j + \numConstant$, with
	$\clock_i \in \ClockSet$,
	$\param_j \in \ParamSet$,
	and
	$\alpha_i, \beta_j, \numConstant \in \setZ$.
A \emph{constraint}~$\Constr$ (\ie{} a convex polyhedron\LongVersion{\footnote{%
	Strictly speaking, we manipulate \emph{polytopes}, while polyhedra refer to 3-dimensional polytopes.
	However, for sake of consistency with the parametric timed model checking literature, and with the Parma polyhedra library (among others), we refer to these geometric objects as \emph{polyhedra}.
}}) over $\ClockSet \cup \ParamSet$ is a conjunction of inequalities of the form $\lterm \compOp 0$, where $\lterm$ is a linear term.
\LongVersion{

}%
Given a parameter valuation~$\pval$, $\valuate{\Constr}{\pval}$ denotes the constraint over~$\ClockSet$ obtained by replacing each parameter~$\param$ in~$\Constr$ with~$\pval(\param)$.
Likewise, given a clock valuation~$\clockval$, $\valuate{\valuate{\Constr}{\pval}}{\clockval}$ denotes the expression obtained by replacing each clock~$\clock$ in~$\valuate{\Constr}{\pval}$ with~$\clockval(\clock)$.
We write $\clockval \models \valuate{\Constr}{\pval}$ whenever $\valuate{\valuate{\Constr}{\pval}}{\clockval}$ evaluates to true.
We say that %
$\pval$ \emph{satisfies}~$\Constr$,
denoted by $\pval \models \Constr$,
if the set of clock valuations satisfying~$\valuate{\Constr}{\pval}$ is nonempty.
We say that $\Constr$ is \emph{satisfiable} if $\exists \clockval, \pval \text{ s.t.\ } \clockval \models \valuate{\Constr}{\pval}$.
\LongVersion{

}%
We define the \emph{time elapsing} of~$\Constr$, denoted by $\timelapse{\Constr}$, as the constraint over $\ClockSet$ and $\ParamSet$ obtained from~$\Constr$ by delaying all clocks by an arbitrary amount of time.
That is,
\(\clockval' \models \valuate{\timelapse{\Constr}}{\pval} \text{ if } \exists \clockval : \ClockSet \to \setRgeqzero, \exists \numConstant \in \setRgeqzero \text { s.t.\ } \clockval \models \valuate{\Constr}{\pval} \land \clockval' = \clockval + \numConstant \text{.}\)
Given $\resets \subseteq \ClockSet$, we define the \emph{reset} of~$\Constr$, denoted by $\reset{\Constr}{\resets}$, as the constraint obtained from~$\Constr$ by resetting the clocks in~$\resets$ to $0$, and keeping the other clocks unchanged.
That is,
\[\clockval' \models \valuate{\reset{\Constr}{\resets}}{\pval} \text{ if } \exists \clockval : \ClockSet \to \setRgeqzero \text { s.t.\ } \clockval \models \valuate{\Constr}{\pval} \land \forall \clock \in \ClockSet
	\left \{ \begin{array}{ll}
		 \clockval'(\clock) = 0 & \text{if } \clock \in \resets\\
		 \clockval'(\clock) = \clockval(\clock) & \text{otherwise.}
	\end{array} \right .\]
We denote by $\projectP{\Constr}$ the projection of~$\Constr$ onto~$\ParamSet$, \ie{} obtained by eliminating the variables not in~$\ParamSet$ (\eg{} using Fourier-Motzkin~\cite{Schrijver86}).

\begin{definition}[Symbolic state]
	A symbolic state is a pair $(\loc, \Constr)$ where $\loc \in \LocSet$ is a location, and $\Constr$ its associated parametric zone.
\end{definition}
\begin{definition}[Symbolic semantics]\label{def:PTA:symbolic}
	Given a PTA $\PTA = \PTAprivextend$, %
	the symbolic semantics of~$\PTA$ is the labeled transition system called \emph{parametric zone graph}
	$ \PZG{\PTA} = ( \EdgeSet, \SymbStateSet, \symbstateinit, \SymbTransitions )$, with
	\begin{itemize}
		\item $\SymbStateSet = \{ (\loc, \Constr) \mid \Constr \subseteq \invariant(\loc) \}$, %
		\LongVersion{\item }$\symbstateinit = \big(\locinit, \timelapse{(\bigwedge_{1 \leq i\leq\ClockCard}\clock_i=0)} \land \invariant(\loc_0) \big)$,
				and
		\item $\big((\loc, \Constr), \edge, (\loc', \Constr')\big) \in {\SymbTransitions}$ if $\edge = (\loc,\guard,\action,\resets,\loc') \in \EdgeSet$ and
			\[\Constr' = \timelapse{\big(\reset{(\Constr \land \guard)}{\resets}\land \invariant(\loc')\big )} \land \invariant(\loc') \text{\ with $\Constr'$ satisfiable.}\]
	\end{itemize}

\end{definition}

That is, in the parametric zone graph, nodes are symbolic states, and arcs are labeled by \emph{edges} of the original PTA.

\subsection{Reachability synthesis}

We use reachability synthesis to solve the problems defined in \cref{section:opacity}.
This procedure, called \EFsynth{}, takes as input a PTA~$\PTA$ and a set of target locations~$\LocsTarget$, and attempts to synthesize all parameter valuations~$\pval$ for which~$\LocsTarget$ is reachable in~$\valuate{\PTA}{\pval}$.
$\EFsynth(\PTA, \LocsTarget)$ was formalized in \eg{} \cite{JLR15} and is a procedure that may not terminate, but that computes an exact result (sound and complete) if it terminates.
\LongVersion{%
\EFsynth{} traverses the parametric zone graph of~$\PTA$.
}
\subsection{Execution-time opacity problems\ShortVersion{~\cite{ALLMS23}}}\label{section:opacity}
\LongVersion{%
	We recall here the notion of execution-time opacity~\cite{ALMS22,ALLMS23}.
	This form of opacity is such that the observation is limited to the time to reach a designated location.
	This section recalls relevant definitions from~\cite{ALMS22,ALLMS23}.
}

\LongVersion{%
\subsubsection{Defining the execution times}\label{sec:opacity:preliminaries:PrivPubVisit} %
}%

Given a \TAtext{}~$\valuate{\PTA}{\pval}$ and a run~$\varrun$, we say that $\locpriv$ is \emph{visited on the way to~$\locfinal$ in~$\varrun$} if $\varrun$ is of the form
\( (\loci{0}, \clockval_0), (d_0, \edgei{0}), (\loci{1}, \clockval_1), \cdots, (\loci{m}, \clockval_m), (d_m, \edgei{m}), \cdots (\loci{n}, \clockval_n)\)
\noindent{}for some~$m,n \in \setN$ such that $\loci{m} = \locpriv$, $\loci{n} = \locfinal$ and $\forall 0 \leq i \leq n-1, \loci{i} \neq \locfinal$.
We denote by $\PrivVisit{\valuate{\PTA}{\pval}}$ the set of those runs, and refer to them as \emph{private} runs.
We denote by $\PrivDurVisit{\valuate{\PTA}{\pval}}$ the set of all the durations of these runs.

Conversely, we say that
$\locpriv$ is \emph{avoided on the way to~$\locfinal$ in~$\varrun$}
if $\varrun$ is of the form
\((\loci{0}, \clockval_0), (d_0, \edgei{0}), (\loci{1}, \clockval_1), \cdots, (\loci{n}, \clockval_n )\)
\noindent{}with $\loci{n} = \locfinal$ and $\forall 0 \leq i < n, \loci{i} \notin \set{\locpriv,\locfinal}$.
We denote the set of those runs by~$\PubVisit{\valuate{\PTA}{\pval}}$, referring to them as \emph{public} runs,
and by $\PubDurVisit{\valuate{\PTA}{\pval}}$ the set of all the durations of these public runs.
\LongVersion{%

}%
Therefore, $\PrivDurVisit{\valuate{\PTA}{\pval}}$ (resp.\ $\PubDurVisit{\valuate{\PTA}{\pval}}$) is the set of all the durations of the runs for which $\locpriv$ is visited (resp.\ avoided) on the way to~$\locfinal$.
\LongVersion{%

}%
These concepts can be seen as the set of \execTimesText{} from the initial location~$\locinit$ to the final location $\locfinal$ while visiting (resp.\ not visiting) a private location~$\locpriv$.
Observe that, from the definition of the duration of a run%
, this ``\execTimeText{}'' does not include the time spent in~$\locfinal$.

\LongVersion{%
\subsubsection{Defining execution-time opacity}
}

We now recall formally the concept of ``\opacityText{} for a set of durations (or \execTimesText{}) $\execTimes$'': a system is \emph{\opaqueText{} for \execTimesText{}~$\execTimes$} whenever, for any duration in~$\execTimes$, it is not possible to deduce whether the system visited~$\locpriv$ or not.
\begin{definition}[\Acf*{opacity} for $\execTimes$]\label{def:opacity:TOSEM:ET-opacity}
	Given a \TAtext{}~$\valuate{\PTA}{\pval}$ and a set of \execTimesText{}~$\execTimes$,
	we say that $\valuate{\PTA}{\pval}$ is \emph{\opaqueTextdef{} for \execTimesText{}~$\execTimes$}
	if $\execTimes \subseteq (\PrivDurVisit{\valuate{\PTA}{\pval}} \cap \PubDurVisit{\valuate{\PTA}{\pval}})$.
\end{definition}

In the following, we will be interested in the \emph{existence} of such an \execTimeText{}.
We say that a \TAtext{} is \existentialOpaqueText{} if it is \opaqueText{} for a non-empty set of \execTimesText{}.

\begin{definition}[\existentialOpacityTextForSec{}]\label{def:opacity:TOSEM:exist-ET-opacity}
	A \TAtext{}~$\valuate{\PTA}{\pval}$ is \emph{\existentialOpaqueText{}}
	if $(\PrivDurVisit{\valuate{\PTA}{\pval}} \cap \PubDurVisit{\valuate{\PTA}{\pval}}) \neq \emptyset$.
\end{definition}

\LongVersion{%
If one does not have the ability to tune the system (\ie{} change internal delays, or add some \stylecode{Thread.sleep()} statements in a program), one may be first interested in knowing whether the system is \opaqueText{} for all \execTimesText{}. %
In other words, }%
\ShortVersion{In addition, }%
a system is \emph{\fullOpaqueText{}} if, for any possible measured \execTimeText{}, an attacker is not able to deduce whether~$\locpriv$ was visited or not.

\begin{definition}[\fullOpacityTextForSec{}]\label{def:opacity:TOSEM:full-ET-opacity}
	A \TAtext{}~$\valuate{\PTA}{\pval}$ is \emph{\fullOpaqueText{}}
	if $\PrivDurVisit{\valuate{\PTA}{\pval}} = \PubDurVisit{\valuate{\PTA}{\pval}}$.
\end{definition}

\LongVersion{%
That is, a system is \fullOpaqueText{} if, for any \execTimeText{} ~$\paramd$, a run of duration~$\paramd$ reaches~$\locfinal$ after visiting~$\locpriv$ iff another run of duration~$\paramd$ reaches~$\locfinal$ without visiting~$\locpriv$.
}
\begin{example}
	Consider again the PTA~$\PTA$ in \cref{figure:example-PTA}.
	Let $\pval$ \suchthat{} $\pval(\parami{1}) = 1$ and $\pval(\parami{2}) = 4$.
	Then $\pval(\PTA)$ is \existentialOpaqueText{} since there is at least one \execTimeText{} for which $\pval(\PTA)$ is \opaqueText{}.
	Here, $\pval(\PTA)$ is \opaqueText{} for \execTimesText{} $[1,3]$.
	However, $\pval(\PTA)$ is not \fullOpaqueText{} since there is at least one \execTimeText{} for which $\pval(\PTA)$ is not \opaqueText{}.
	Here, $\pval(\PTA)$ is not \opaqueText{} for \execTimesText{} $[0,1)$ (which can only occur on a public run) and for \execTimesText{} $(3,4]$ (which can only occur on a private run).
\end{example}

\LongVersion{%
\subsubsection{Decision and computation problems}\label{sec:opacity:TOSEM:problems}
}

\LongVersion{%
\paragraph{Decision problems} %
}

Let us consider the following decision problems\LongVersion{, \ie{} the problem of checking the \emph{emptiness} of the set of parameter valuations guaranteeing \existentialOpacityText{} (hereafter $\problemTOE$) and its counterpart for \fullOpacityText{} (hereafter $\problemFTOE$)}:

\defProblem
{\existentialOpacityParamEmptinessProblem{} ($\problemTOE$)}
{A \PTAtext{}~$\PTA$
}
{Decide the emptiness of the set of \LongVersion{parameter }valuations $\pval$
\suchthat{} $\valuate{\PTA}{\pval}$ is \existentialOpaqueText{}.}

\LongVersion{%
	The negation of the \existentialOpacityParamEmptinessProblem{} consists in deciding whether there exists at least one parameter valuation for which $\valuate{\PTA}{\pval}$ is \existentialOpaqueText{}.
}

\defProblem
{\FullOpacityParamEmptinessProblem{} ($\problemFTOE$)}
{A \PTAtext{}~$\PTA$
}
{Decide the emptiness of the set of \LongVersion{parameter }valuations $\pval$ \suchthat{} $\valuate{\PTA}{\pval}$ is \fullOpaqueText{}.}

\LongVersion{%
	Equivalently, we are interested in deciding whether there exists at least one parameter valuation for which $\valuate{\PTA}{\pval}$ is \fullOpaqueText{}.
}

\LongVersion{%
\paragraph{Synthesis problems} %
}

The synthesis counterpart allows for a higher-level problem aiming at synthesizing (ideally the entire set of) parameter valuations~$\pval$ for which $\valuate{\PTA}{\pval}$ is \existentialOpaqueText{} or \fullOpaqueText{}.

\defProblem
{\existentialOpacityParamSynthesisProblem{} ($\problemTOS$)}
{A \PTAtext{}~$\PTA$
}
{Synthesize the set of all \LongVersion{parameter }valuations $\pval$ \suchthat{} $\valuate{\PTA}{\pval}$ is \existentialOpaqueText{}.}

\defProblem
{\FullOpacityParamSynthesisProblem{} ($\problemFTOS$)}
{A \PTAtext{}~$\PTA$
}
{Synthesize the set of all \LongVersion{parameter }valuations $\pval$ \suchthat{} $\valuate{\PTA}{\pval}$ is \fullOpaqueText{}.}

\section{A parametric execution times-based semi-algorithm for \problemTOS{} and \problemFTOS{}}\label{section:PETS}

One of our main results is the proof that both \problemTOS{} and \problemFTOS{} can be deduced from set operations on two sets representing respectively all the durations and parameter valuations of the runs for which $\locpriv$ is reached (resp.\ avoided) on the way to~$\locfinal$.
Those sets can be seen as a parametrized version of $\PrivDurVisit{\valuate{\PTA}{\pval}}$ and $\PubDurVisit{\valuate{\PTA}{\pval}}$.
In order to compute such sets, we propose here the novel notion of parametric execution times.
(Note that our partial solution for \textPET{} construction and  semi-algorithms for \problemTOS{} and \problemFTOS{} work perfectly for \emph{rational}-valued parameters too, and that they are not restricted to 1-clock PTAs.)

\subsection{Parametric execution times}\label{ss:PETS}

The parametric execution times (\textPET{}) are the parameter valuations and execution times of the runs to~$\locfinal$. %

\begin{definition}
\label{defPET}
	Given a PTA~$\PTA$ with final location~$\locfinal$, the \emph{parametric execution times} of~$\PTA$ are defined as
	$\PET{\PTA} = \{ (\pval, d) \mid \exists \varrun \text{ in } \valuate{\PTA}{\pval} \text{ such that } d = \duration(\varrun)\ \land\ \varrun \text{ is of the form } (\loc_0, \clockval_0), (d_0, \edge_0), \cdots, (\loc_n, \clockval_n)$ for some $n \in \setN$ such that $\loc_n = \locfinal$ and $\forall 0 \leq i \leq n-1, \loc_i \neq \locfinal \}$.
\end{definition}

By definition, we only consider paths up to the point where $\locfinal$ is reached, meaning that execution times do not include the time elapsed in $\locfinal$, and that runs that reach $\locfinal$ more than once are only considered up to their first visit of~$\locfinal$.

\begin{example}
	Consider again the PTA~$\PTA$ in \cref{figure:example-PTA}.
	Then $\PET{\PTA} $ is
	$(d \leq 3 \land \parami{1} \geq 0 \land \parami{2} \geq 0)    \lor   (0 \leq \parami{1} \leq 3 \land \parami{1} \leq d \leq \parami{2})$.
\end{example}
\subsubsection{Partial solution}

Synthesizing parametric execution times is in fact equivalent to a reachability synthesis where the PTA is enriched (in particular by adding a clock measuring the total execution time).

\begin{restatable}{proposition}{propGeneralPETS}
\label{proposition:GeneralPETS}
	Let $\PTA$ be a PTA, and $\locfinal$ the final location of~$\PTA$.\\
	Let $\PTA'$ be a copy of~$\PTA$ \suchthat{}:
	\begin{itemize}
	\item a clock $\clockabs$ is added and initialized at $0$ (it does not occur in any guard or reset);
	\item a parameter $\paramabs$ is added;
	\item $\locfinal$ is made \emph{urgent} (\ie{} time is not allowed to pass in $\locfinal$), all outgoing edges from $\locfinal$ are pruned and a guard $\clockabs = \paramabs$ is added to all incoming edges to $\locfinal$.
	\end{itemize}
	Then, $\PET{\PTA} = \EFsynth(\PTA', \{\locfinal\})$.
\end{restatable}
\begin{example}
	Consider again the PTA~$\PTA$ in \cref{figure:example-PTA}.
	Then $\PTA'$ is given in \cref{figure:example-PTA-transformed-PET}.
\end{example}

As per \cref{prop:EFsynth} in \cref{appendix:prop:EFsynth}, there exist \LongVersion{sound and correct }semi-algorithms for reachability synthesis, and hence for the \textPET{} synthesis problem---although they do not guarantee termination.

\subsection{\problemTOS{} and \problemFTOS{} problems}\label{ss:problems}

Now, we detail how the \textPET{} can be used to compute the solution to both \problemTOS{} and~\problemFTOS{}.
To do so, we will go through a (larger) intermediate problem: the synthesis of both parameter valuations $\pval$ \emph{and} execution times for which $\valuate{\PTA}{\pval}$ is \opaqueText{}.

\defProblem
{\existentialOpacitySynthesisProblem{} ($\problemdTOS$)}
{A \PTAtext{}~$\PTA$
}
{
Synthesize the set of parameter valuations~$\pval$ and execution times~$d$ \suchthat{} $\valuate{\PTA}{\pval}$ is \existentialOpaqueText{} and $\valuate{\PTA}{\pval}$ is \opaqueText{} for execution time~$d$.
}

\defProblem
{\FullOpacitySynthesisProblem{} ($\problemdFTOS$)}
{A \PTAtext{}~$\PTA$
}
{%
Synthesize the set of parameter valuations~$\pval$ and execution times~$d$ \suchthat{} $\valuate{\PTA}{\pval}$ is \fullOpaqueText{} and 
 $d$ is the set of durations of all runs in $\valuate{\PTA}{\pval}$.
}

First, given a PTA~$\PTA$ and two locations $\locfinal$ and $\locpriv$ of~$\PTA$, let us formally define both sets representing respectively all the durations and parameter valuations of the runs for which $\locpriv$ is reached (resp.\ avoided) on the way to~$\locfinal$.

Let $\PTA^{\locpriv}_{\locfinal}$ be a copy of~$\PTA$ \suchthat{}:
\begin{oneenumerate}
	\item a Boolean variable\footnote{Which is a convenient syntactic sugar for doubling the number of locations.} $b$ is added and initialized to \false,
	\item $b$ is set to \true\ on all incoming edges to $\locpriv$,
	\item a guard $b = \true$ is added to all incoming edges to $\locfinal$.
\end{oneenumerate}%
The PTA $\PTA^{\locpriv}_{\locfinal}$ contains all runs of~$\PTA$ for which $\locpriv$ is reached on the way to~$\locfinal$, and $\PET{\PTA^{\locpriv}_{\locfinal}}$ contains the durations and parameter valuations of those runs.

Let $\PTA^{\neg \locpriv}_{\locfinal}$ be a copy of~$\PTA$ \suchthat{} all incoming and outgoing edges to and from $\locpriv$ are pruned.
The PTA $\PTA^{\neg \locpriv}_{\locfinal}$ contains all runs of~$\PTA$ for which $\locpriv$ is avoided on the way to~$\locfinal$, and $\PET{\PTA^{\neg \locpriv}_{\locfinal}}$ contains the durations and parameter valuations of those runs.

\begin{example}
	Consider again the PTA~$\PTA$ in \cref{figure:example-PTA}.
	Then $\PTA^{\locpriv}_{\locfinal}$ is given in \cref{figure:example-PTA-transformed-priv}, and $\PTA^{\neg \locpriv}_{\locfinal}$ is given in \cref{figure:example-PTA-transformed-pub}.
\end{example}

\LongVersion{
Let us consider the inclusion relations between $\PET{\PTA^{\locpriv}_{\locfinal}}$, $\PET{\PTA^{\neg \locpriv}_{\locfinal}}$, $\problemTOS(\PTA)$ and $\problemFTOS(\PTA)$ (\cref{figure:TOsets}).
By construction, the union of both \textPET{}s corresponds to $\PET{\PTA}$, while their intersection corresponds to $\problemTOS(\PTA)$.
As one can see in the figure, $\problemFTOS(\PTA) \subseteq \problemTOS(\PTA)$ and cannot be expressed directly with \textPET{}s.

 \begin{figure} [htb]
 	\centering
 	\small
 		\begin{tikzpicture}[yscale=.5]
 			\tikzstyle{openzone1} = [fill=USPNsable!20!white, draw=none]
 			\tikzstyle{openzone2} = [fill=USPNsable!40!white, draw=none]
 			\tikzstyle{openzone3} = [fill=USPNsable!60!white, draw=black]

 			\draw[openzone1] (0, -0.25) -- (0, 4.75) -- (5, 4.75)-- (5, -0.25) -- cycle;
 			\draw[openzone1] (1, 1.25) -- (1, 6.25) -- (6, 6.25)-- (6, 1.25) -- cycle;
 			\draw[openzone2] (1, 1.25) -- (1, 4.75) -- (5, 4.75)-- (5, 1.25) -- cycle;
 			\draw[openzone3] (1.5, 2) -- (1.5, 3.5) -- (4.5, 3.5)-- (4.5, 2) -- cycle;
			\draw[densely dotted, line width = 0.3mm, black] (0, -0.25) -- (0, 4.75) -- (5, 4.75)-- (5, -0.25) -- cycle;
			\draw[dashed, line width = 0.3mm, black] (1, 1.25) -- (1, 6.25) -- (6, 6.25)-- (6, 1.25) -- cycle;

 			\node at (4.5,5.5) {$\PET{\PTA^{\locpriv}_{\locfinal}}$};
 			\node at (1.5,0.5) {$\PET{\PTA^{\neg \locpriv}_{\locfinal}}$};
 			\node at (2.5,4) {\problemTOS(\PTA)};
 			\node at (3,2.75) {\problemFTOS(\PTA)};
 		\end{tikzpicture}
 	\caption{Inclusion relations between the sets discussed in this section. Each set is a union of convex polyhedra.}
 	\label{figure:TOsets}
 \end{figure}
}

\begin{restatable}{proposition}{propTOS}\label{TOS}
Given a PTA~$\PTA$, we have:
 	\(\problemdTOS(\PTA) = \PET{\PTA^{\locpriv}_{\locfinal}} \cap \PET{\PTA^{\neg \locpriv}_{\locfinal}}\text{.}\)
\end{restatable}
\begin{example}
	Consider again the PTA~$\PTA$ in \cref{figure:example-PTA}.
	Then $\PET{\PTA^{\locpriv}_{\locfinal}}$ is $\parami{1} \leq \paramabs \leq \parami{2} \land 0 \leq \parami{1} \leq 3$.
	Moreover, $\PET{\PTA^{\neg \locpriv}_{\locfinal}}$ is $0 \leq \paramabs \leq 3 \land \parami{1} \geq 0 \land \parami{2} \geq 0$.
	Hence, $\problemdTOS(\PTA)$ is $0 \leq \parami{1} \leq \paramabs \leq \parami{2} \land \paramabs \leq 3$.
\end{example}

In order to compute $\problemdFTOS(\PTA)$, we need to remove from $\problemdTOS(\PTA)$ all parameter valuations~$\pval$ \suchthat{} there is at least one run to~$\locfinal$ in $\valuate{\PTA}{\pval}$ whose duration is not\LongVersion{ in~$\execTimes_\pval$.} 
in the set of execution times for which $\valuate{\PTA}{\pval}$ is \opaqueText{}.
Parameter valuations and durations of such runs are included in $\PET{\PTA} \setminus \problemdTOS(\PTA)$, which is also the difference between $\PET{\PTA^{\locpriv}_{\locfinal}}$ and $\PET{\PTA^{\neg \locpriv}_{\locfinal}}$.
We note that difference as
\[\Diff(\PTA) = \big(\PET{\PTA^{\locpriv}_{\locfinal}} \cup \PET{\PTA^{\neg \locpriv}_{\locfinal}} \big) \setminus \big(\PET{\PTA^{\locpriv}_{\locfinal}} \cap \PET{\PTA^{\neg \locpriv}_{\locfinal}}\big)\]
$\Diff(\PTA)$ is made of a union of convex polyhedra $\Constr$ over $\ParamSet$ (\ie{} the parameters of~$\PTA$) and~$\paramabs$, which is the duration of runs.
The parameter values in those polyhedra are the ones we do not want to see in $\problemdFTOS(\PTA)$.
Our solution thus consists in removing from $\problemdTOS(\PTA)$ the values of $\ParamSet$ in $\Diff(\PTA)$.

\begin{restatable}{proposition}{propFullTOS}\label{fullTOS}
Given a PTA~$\PTA$ with parameter set $\ParamSet$:
	\(\problemdFTOS(\PTA) = \problemdTOS(\PTA) \setminus \projectP{\Diff(\PTA)}\text{.}\)
\LongVersion{ where $\projectP{\Diff(\PTA)}$ is the union of convex polyhedra obtained by rewriting each constraint $\Constr \in \Diff(\PTA)$ by $\projectP{\Constr}$.}
\end{restatable}
\begin{example}
	Consider again the PTA~$\PTA$ in \cref{figure:example-PTA}.
	Whe have $\Diff(\PTA)$ is $(0 \leq \parami{1} \leq 3 < \paramabs \leq \parami{2}) \lor (0 \leq \paramabs \leq 3 \land \paramabs < \parami{1} \land \parami{2} \geq 0) \lor (0 \leq \parami{2} < \paramabs \leq 3 \land \parami{1} \geq 0)$.
	Then $\projectP{\Diff(\PTA)}$ is $(0 \leq \parami{1} \leq 3 < \parami{2}) \lor (0 < \parami{1} \land \parami{2} \geq 0) \lor (0 \leq \parami{2} < 3 \land \parami{1} \geq 0)$.
	Hence, $\problemdFTOS(\PTA)$ is $\parami{1} = 0 \leq \paramabs \leq \parami{2} = 3$.
	
\end{example}

Finally, obtaining $\problemTOS(\PTA)$ and $\problemFTOS(\PTA)$ is trivial since, by definition, $\problemTOS(\PTA) = \projectP{(\problemdTOS(\PTA))}$ and $\problemFTOS(\PTA) = \projectP{(\problemdFTOS(\PTA))}$.

\begin{example}
Consider again the PTA~$\PTA$ in \cref{figure:example-PTA}.
Then $\problemTOS(\PTA)$ is $0 \leq \parami{1} \leq \parami{2} \land \parami{1} \leq 3$.
And $\problemFTOS(\PTA)$ is $\parami{1} = 0 \land \parami{2} = 3$.
\end{example}

\LongVersion{
\subsubsection{Illustration}

Let us illustrate on a simple example how \textPET{} are used to solve \problemTOS{} and \problemFTOS{} problems, as described above.
Let us assume we have a PTA~$\PTA$ with one parameter~$p$ on which we successfully computed $\PET{\PTA^{\locpriv}_{\locfinal}}$ and $\PET{\PTA^{\neg \locpriv}_{\locfinal}}$:

$
\begin{array}{cclcccl}
\PET{\PTA^{\locpriv}_{\locfinal}}	& = &		\{ p = 1 \land 0 \leq \paramabs \leq 3 \}	&\ $and$ & \PET{\PTA^{\neg \locpriv}_{\locfinal}} & = &		\{ p = 1 \land 5 \leq \paramabs \leq 10 \} \\
		& \cup &	\{ p = 2 \land 0 \leq \paramabs \leq 5 \}	&&			& \cup &	\{ p = 2 \land 3 \leq \paramabs \leq 10 \} \\
		& \cup &	\{ p = 3 \land 0 \leq \paramabs \leq 5 \}	& &				& \cup &	\{ p = 3 \land 0 \leq \paramabs \leq 5 \}
\end{array}
$

\smallskip

We then compute \Diff(\PTA) and its projection onto $\ParamSet$:

$
\begin{array}{cclcccl}
\Diff(\PTA) & = &	\{ p = 1 \land 0 \leq \paramabs \leq 3 \}	&\ $and$ & \projectP{\Diff(\PTA)} & = &		\{ p = 1 \} \\
		& \cup &	\{ p = 1 \land 5 \leq \paramabs \leq 10 \}	&&				& \cup &	\{ p = 2 \} \\
		& \cup &	\{ p = 2 \land 0 \leq \paramabs < 3 \}	& &					&  &	\\
		& \cup &	\{ p = 2 \land 5 < \paramabs \leq 10 \}	& &					&  &
\end{array}
$

\smallskip

We then have $\problemdTOS(\PTA) = \PET{\PTA^{\locpriv}_{\locfinal}} \cap \PET{\PTA^{\neg \locpriv}_{\locfinal}} = \{ p = 2 \land 3 \leq \paramabs \leq 5 \} \cup \{ p = 3 \land 0 \leq \paramabs \leq 5 \}$.
And $\problemdFTOS(\PTA) = \problemTOS(\PTA) \setminus \projectP{\Diff(\PTA)} = \{ p = 3 \land 0 \leq \paramabs \leq 5 \}$.
}

\subsubsection{On correctness and termination}

We described here a method for computing \problemTOS(\PTA) and \problemFTOS(\PTA) for a PTA, that produces an exact (sound and complete) result if it terminates.
It relies on the \textPET{} of two subsets of the PTA, the computation of which requires enrichment with one clock and one parameter.
If they can be computed, those \textPET{} take the form of a finite union of convex polyhedra, on which are then applied the union, intersection, difference and projection set operations --- that are known to be decidable in this context.
Thus the actual termination of the whole semi-algorithm relies on the reachability synthesis of two $(n+1, 0, m+1)$-PTAs.
Reachability synthesis is known to be effectively computable for $(1, 0, m)$-PTAs~\cite{ALM20}, and cannot be achieved for PTAs with 3 parametric clocks or more due to the undecidability of the reachability emptiness problem~\cite{AHV93}.
For the semi-algorithm we proposed here for \problemTOS{} and \problemFTOS{} problems, we therefore do not have any guarantees of termination, even with only one parametric clock (due to the additional clock $\clockabs$), although this might change depending on future results regarding the decidability of reachability synthesis for PTAs with 2 parametric clocks (a first decidability result for the emptiness only was proved for $(2,\arbitrarilyMany,1)$-PTAs over discrete time~\cite{GH21}).

\section{Decidability and undecidability of \problemFTOE{} for 1-clock-PTAs\LongVersion{ with integer parameters over dense time}}\label{section:TOS}

\LongVersion{%
Recall that we denote by $(pc, npc, p)$-PTAs the class of PTAs where $pc$ is the number of parametric clocks (\ie{} compared at least once in a guard to a parameter), $npc$ the number of non-parametric clocks and $p$ the number of parameters (``$\arbitrarilyMany$'' denotes ``any number).
}

In this section, we:
\begin{enumerate}
	\item propose a method to compute potentially infinite \textPET{} on $(1, 0, \arbitrarilyMany)$-PTAs, \ie{} PTAs with 1 parametric clock and arbitrarily many parameters\LongVersion{ (integer-valued in this work, but which could also be rational-valued technically)} (\cref{ss:infinitePETS});
	\item prove decidability of the \problemFTOE{} problem for $(1, 0, 1)$-PTAs\LongVersion{ (over dense time and an integer-valued parameter)}, by rewriting infinite \textPET{} in a variant of Presburger arithmetic (\cref{ss:PresburgerFTOE});
	\item prove undecidability of the \problemFTOE{} problem for $(1, 0, \arbitrarilyMany)$-PTAs\LongVersion{ with integer-valued parameters over dense time (and thus for rational-valued $(1, 0, \arbitrarilyMany)$-PTAs as well)} (\cref{ss:PresburgerFTOE}).
\end{enumerate}

\subsection{Encoding infinite \textPET{} for $(1, 0, \arbitrarilyMany)$-PTAs}\label{ss:infinitePETS}
Given a PTA~$\PTA$ with exactly 1 clock, \LongVersion{and given a location $\locfinal$ of~$\PTA$, }the goal of the method described here is to guarantee termination of the computation of $\PET{\PTA}$ with an exact result.
If the partial solution given in \cref{ss:PETS} is applied, it amounts to a reachability synthesis on a PTA with 2 clocks, without guarantee of termination.
The gist of this method is a form of divide and conquer, where we solve sub-problems, specifically reachability synthesis on sub-parts of~$\PTA$ without adding an additional clock.
The first step consists of building some reset-free PTAs, each representing a meaningful subset of the paths joining two given locations in $\PTA$.
$\PET{\PTA}$ is then obtained by combining the results of reachability synthesis performed on those reset-free PTAs.
The result is encoded in a (finite) regular expression that represents an infinite union of convex polyhedra.
Note that this method works perfectly for rational-valued parameters.

\subsubsection{Defining the set of reset-free PTAs}

Each of the PTAs we build describes parts of the behavior between two locations.
More precisely, they represent all the possible paths such that clock resets may occur only on the last transition of the path.
We first define the set of locations that we may need based on whether they are initial, final, or reached by a transition associated to a reset.

\begin{definition}[Final-reset paths $\Rfp(\PTA, \locfinal)$]
\label{def_rfp}
	Let $\PTA$ be a 1-clock PTA, $\locinit$ its initial location and $\locfinal$ a location of~$\PTA$.
	We define as $\Rfp(\PTA, \locfinal)$ the set of pairs of locations \suchthat{} $\forall (\loc_i, \loc_j) \in \Rfp(\PTA, \locfinal)$
\begin{itemize}
\item %
 $\loc_i = \locinit$, or $\loc_i \neq \locfinal$ and there is a clock reset on an incoming edge to $\loc_i$, 
\item %
$\loc_j = \locfinal$, or there is a clock reset on an incoming edge to $\loc_j$.%
\end{itemize}
\end{definition}

For each pair of states $(\loc_i,\loc_j)$ as defined above, we build a reset-free PTA. 
If the target state $\loc_j$ is not final (which is a special case), the reset-free PTA models every path going from $\loc_i$ to $\loc_j$ and that ends with a reset on its last step.
In particular, this ensures that $\loc_j$ is reached with clock valuation~$0$.

\begin{definition}[Reset-free PTA $\PTA(\loc_i,\loc_j)$]\label{sub-automata}
	Let $\PTA$ be a 1-clock PTA, $\clock$ its unique clock, and $\loc_i$, $\loc_j$ two locations in $\PTA$.
	We define as $\PTA(\loc_i,\loc_j)$ the reset-free PTA obtained from a copy of~$\PTA$ by:
	\begin{enumerate}
	\item creating a duplicate $\loc_j'$ of $\loc_j$;
	\item for all incoming edges $(\loc,\guard,\action,\resets,\loc_j)$ where $\resets = \emptyset$, removing $(\loc,\guard,\action,\resets,\loc_j)$ and adding an incoming edge $(\loc,\guard,\action,\resets,\loc'_j)$;
	\item if $\loc_j \neq \locfinal$, then for all outgoing edges $(\loc_j,\guard,\action,\resets,\loc)$,  removing $(\loc_j,\guard,\action,\resets,\loc)$ and adding an outgoing edge $(\loc'_j,\guard,\action,\resets,\loc)$,

	else, making $\loc_j'$ urgent and adding an edge $(\loc_j',\true,\epsilon,\emptyset,\loc_j)$;
	\item removing any upper bound invariant on $\loc_j$ and making it urgent;
	\item  if $\loc_i \neq \loc_j$, setting $\loc_i$ as the initial location,

	else, setting $\loc_j'$ as the initial location;
	\item removing any clock reset on incoming edges to $\loc_j$ and pruning all other edges featuring a clock reset, and all outgoing edges from $\locfinal$;
	\item adding a parameter $\paramabs$, and a guard $\clock = \paramabs$ to all incoming edges to $\loc_j$;
	\end{enumerate}
\end{definition}

We will show next how the reachability synthesis of those reset-free PTAs corresponds to fragments of the runs that are considered in $\PET{\PTA}$.
\LongVersion{The following two proposition will be needed for that demonstration.}
For simplification, given $\PTA$ a 1-clock PTA, and $\loc_i$, $\loc_j$ two locations of~$\PTA$, we now note $Z_{\loc_i,\loc_j} = \EFsynth(\PTA(\loc_i,\loc_j), \{\loc_j\})$.

\LongVersion{
\begin{restatable}{proposition}{propequivnotlf}
\label{equiv_not_lf}
	Let~$\PTA$ be a 1-clock PTA, and $(\loc_i,\loc_j) \in \Rfp(\PTA, \locfinal)$ such that $\loc_j \neq \locfinal$.
	Then $Z_{\loc_i,\loc_j}$ is equivalent to the synthesis of parameter valuations~$\pval$ and execution times~$\execTimes_\pval$ such that $\execTimes_\pval = \{ d \mid \exists \varrun$ from $(\loc_i,\{\clock = 0\})$ to $\loc_j$ in $\valuate{\PTA}{\pval}$ such that  $d = \duration(\varrun)$, $\locfinal$ is never reached, and $\clock$ is reset on the last edge of $\varrun$ and on this edge only $\}$.
\end{restatable}
\begin{proof}
	See \cref{appendix:proof:equiv_not_lf}.
\end{proof}
\begin{restatable}{proposition}{propequivlf}
\label{equiv_lf}
	Let~$\PTA$ be a 1-clock PTA, and $(\loc_i,\loc_j) \in \Rfp(\PTA, \locfinal)$ such that $\loc_j = \locfinal$.
	Then $Z_{\loc_i,\loc_j}$ is equivalent to the synthesis of parameter valuations~$\pval$ and execution times~$\execTimes_\pval$ such that $\execTimes_\pval = \{ d \mid \exists \varrun$ from $(\loc_i,\{\clock = 0\})$ to $\locfinal$ in $\valuate{\PTA}{\pval}$ such that  $d = \duration(\varrun)$, $\locfinal$ is reached only on the last state of $\varrun$, and $\clock$ may only be reset on the last edge of $\varrun$ $\}$.
\end{restatable}
\begin{proof}
	See \cref{appendix:proof:equiv_lf}.
\end{proof}
}

\subsubsection{Reconstruction of \textPET{} from the reachability synthesis of the reset-free PTAs.}

Given $\PTA$ a 1-clock PTA, and $\locfinal$ a location of~$\PTA$, for all $(\loc_i,\loc_j) \in \Rfp(\PTA, \locfinal)$ we may compute the parametric zone $Z_{\loc_i,\loc_j}$ with guarantee of termination, since the reachability synthesis is decidable on 1-clock PTAs.
Those parametric zones may be used to build the (potentially infinite) \textPET{} of~$\PTA$.
To do so, we first define a (non-parametric, untimed) finite automaton where the states are the locations of~$\PTA$, and the arc between the states $\loc_i$ and $\loc_j$ is labeled by $Z_{\loc_i,\loc_j}$.
We refer to this automaton as the \emph{automaton of the zones} of~$\PTA$.

\begin{definition}[Automaton of the zones]\label{zone_automaton}
	Let $\PTA$ be a 1-clock PTA, $\locinit$ its initial location and $\locfinal$ a location of~$\PTA$.
	We define as $\hat{\PTA}$ the finite automaton such that:
\begin{itemize}
\item The states of $\hat{\PTA}$  are exactly the locations of~$\PTA$;
\item $\locinit$ is initial and $\locfinal$ is final;
\item $\forall (\loc_i,\loc_j) \in \Rfp(\PTA, \locfinal)$, there is a transition from $\loc_i$ to $\loc_j$ labeled by $Z_{\loc_i,\loc_j}$.
\end{itemize}
\end{definition}

We claim that the language $\hat{L}$ of $\hat{\PTA}$ is a representation of the times (along with parameter constraints) to go from $\locinit$ to $\locfinal$ in $\PTA$.
As $\hat{\PTA}$ is a finite automaton, $\hat{L}$ can be represented as a regular expression with three operators: the concatenation ($.$), the alteration ($+$), and the Kleene star ($^*$).
$\PET{\PTA}$ can thus be expressed by redefining those operators with operations on the parametric zones that label edges of $\hat{L}$.

Any parametric zone $Z_{a,b}$ labeling an edge of $\hat{\PTA}$ is of the form $\bigcup_{i}{\Constr_i}$ with $1 \leq i \leq n$ and ${\Constr_i}$ a convex polyhedra.
As per \cref{def:PTA:symbolic}, $\Constr_i$ is a conjunction of inequalities, each of the form $\alpha \paramabs + \sum_{1 \leq i \leq \ParamCard} \beta_i \param_i + \numConstant \compOp 0$, with $\param_i \in \ParamSet$, and $\alpha, \beta_i, \numConstant \in \setZ$.
Note that $\clock$ has been replaced by execution times $\paramabs$, as per \cref{defPET}.
In the following, we denote by $\Constr_i^\paramabs$ all inequalities such that $\alpha \neq 0$ (\ie{} inequalities over $\paramabs$ and possibly some parameters in $\mathbb{P}$), and by $\Constr_i^\mathbb{P}$ all inequalities such that $\alpha = 0$ (\ie{} inequalities strictly over $\mathbb{P}$).
This means that $\Constr_i = \Constr_i^\paramabs \cap \Constr_i^\mathbb{P}$.
For simplification of what follows, we write inequalities in $\Constr_i^\paramabs$ as  $\paramabs \compOp c$ where $c = \frac{\sum_{1 \leq i \leq \ParamCard} \beta_i \param_i + \numConstant}{-\alpha}$.

Given $Z_{a,b} = \bigcup_{i}{\Constr_i}$ and $Z_{c,d} = \bigcup_{j}{\Constr_j}$, we define the operators $\bar{.}$, $\bar{*}$ and $\bar{+}$\ .

Operator $\bar{.}$ is the addition of the time durations and intersection of parameter constraints between two parametric zones.
Formally, $Z_{a,b}\ \bar{.}\ Z_{c,d} = \bigcup_{i*j}{\Constr_{i,j}^\paramabs \cap \Constr_{i,j}^\mathbb{P}}$ such that $\Constr_{i,j}^\mathbb{P} = \Constr_i^\mathbb{P} \cap \Constr_j^\mathbb{P}$, and for all $\paramabs \bowtie c_i \in \Constr_i^\paramabs$ and $\paramabs \bowtie' c_j \in \Constr_j^\paramabs$, if ${\bowtie}, {\bowtie'} \in \{<, \leq , = \}$ or ${\bowtie}, {\bowtie'} \in \{>, \geq  , =\}$, then $\paramabs \bowtie'' c_i + c_j \in \Constr_{i,j}^\paramabs$ with $\bowtie''$ being in the same direction as $\bowtie$ and $\bowtie'$ and is
\begin{itemize}
\item a strict inequality if either $\bowtie$ or $\bowtie'$  is a strict inequality; 
\item an equality if both $\bowtie$ and $\bowtie'$ are equalities;
\item a non-strict inequality otherwise.
\end{itemize}

Operator $\bar{*}$ is the recursive application of $\bar{.}$ on a parametric zone.
Formally, ${Z_{a,b}}^{\bar{*}} = \bigcup_{K\in\setN}{  \{\paramabs = 0\} (\bar{.} Z_{a,b})^K}$ where $(\bar{.} Z_{a,b})$ is repeated $K$ times, with $K$ being any value in $\ensuremath{\mathbb N}$. Note that $\{\paramabs = 0\}$ corresponds to the case where the loop is never taken, and that it is neutral for the $\bar{.}$ operator: $\{\paramabs = 0\} \bar{.} Z_{a,b} = Z_{a,b}$.
Also note that, in practice, $a=b$ whenever we use this operator.
 
Operator $\bar{+}$ is the union of two parametric zones.
Formally, $Z_{a,b} \bar{+} Z_{c,d} = Z_{a,b} \cup Z_{c,d}$.

Note that the result of any of those operations is a union of convex polyhedra of the form $\bigcup_{i}{\Constr_i}$, meaning that these operators can be nested.
Also, this union is infinite whenever operator $\bar{*}$ is present.

\begin{restatable}{proposition}{proplanguageequivalence}
\label{language equivalence}
	Let~$\PTA$ be a 1-clock PTA and $\locfinal$ a location of~$\PTA$.
	Let $\hat{L}$ be the language of the automaton of the zones $\hat{\PTA}$, and $e$ a regular expression describing $\hat{L}$.
	Let $\bar{e}$ be the expression obtained by replacing the $.$, $+$ and $^*$ operators in $e$ respectively by $\bar{.}$, $\bar{+}$ and $^{\bar{*}}$.
	We have $\bar{e} = \PET{\PTA}$.
\end{restatable}

\subsubsection{Summary and illustration of the encoding}
Given a PTA~$\PTA$ with exactly 1 clock, and given a location $\locfinal$ of~$\PTA$, we compute with an exact result an encoding of $\PET{\PTA}$, through the following steps:
\begin{enumerate}
\item compute $\Rfp(\PTA, \locfinal)$, the pairs of locations $(\loc_i,\loc_j)$ such that on some run from initial location to $\locfinal$ there might exists a sub-path from $\loc_i$ to $\loc_j$, such that the clock is reset when entering both locations, but never in between;
\item for each of those pairs, compute the reset-free PTA $\PTA(\loc_i,\loc_j)$, for which reachability synthesis, noted $Z_{\loc_i,\loc_j}$ corresponds to the aforementioned sub-paths;
\item generate the automaton of the zones $\hat{\PTA}$, on which each pair of locations $(\loc_i,\loc_j)$ is connected by a transition labeled with $Z_{\loc_i,\loc_j}$;
\item compute a regular expression for $\hat{\PTA}$, which we proved to be equivalent to $\PET{\PTA}$. Note that computing a regular expression from a finite automaton is decidable and there exists numerous efficient methods for this~\cite{GH15}.
\end{enumerate}

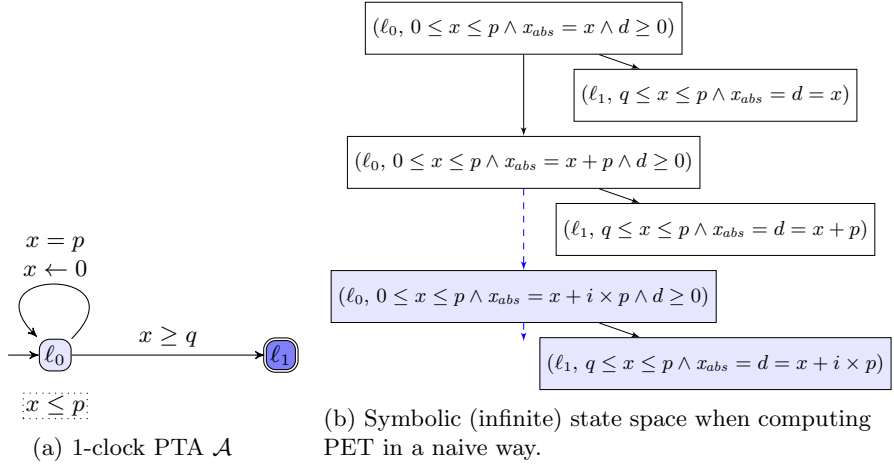
\begin{figure} [tb]
	\centering
	\small
	\begin{subfigure}[b]{0.38\textwidth}
		\centering
		\begin{tikzpicture}[PTA, thin]

			\node[location, initial] at(0,0) (l0) {$\loc_0$};
			\node[location,final] at(3,0) (l1) {$\loc_1$};
			\node [invariant, below] at (0,-0.5) {$x \leq p$};

			\path
				(l0) edge[loop] node [above,align=center] {$x = p$\\$x \leftarrow 0$} (l0)
				(l0) edge[] node {$x \geq q$}  (l1)
			;
		\end{tikzpicture}
		\caption{1-clock PTA~$\PTA$}
		\label{figure:1-clock_pta:pta}
	\end{subfigure}
		\hfill
	\begin{subfigure}[b]{0.6\textwidth}
		\centering
		\scalebox{.85}{
		\begin{tikzpicture}[>=latex', xscale=1.2, yscale=.7,every node/.style={scale=1}]
		\node[symbstate] at (0,0) (c0) {($\loc_0$, $0 \leq x \leq p \land \clockabs = x \land \paramabs \geq 0$)};
		\node[symbstate] at (2.5,-1.5) (c0') {($\loc_1$, $q \leq x \leq p \land \clockabs = \paramabs = x$)};
		\node[symbstate] at (0,-3) (c1) {($\loc_0$, $0 \leq x \leq p \land \clockabs = x + p \land \paramabs \geq 0$)};
		\node[symbstate] at (2.5,-4.5) (c1') {($\loc_1$, $q \leq x \leq p \land \clockabs = \paramabs = x + p$)};
		\node[infinitesymbstate] at (0,-6) (ci) {($\loc_0$, $0 \leq x \leq p \land \clockabs = x + i \times p \land \paramabs \geq 0$)};
		\node[infinitesymbstate] at (2.5,-7.5) (ci') {($\loc_1$, $q \leq x \leq p \land \clockabs = \paramabs = x + i \times p$)};
		\draw[->] (c0) -- (c0');
		\draw[->] (c0) -- (c1);
		\draw[->] (c1) -- (c1');
		\draw[->, blue,dashed] (c1) -- (ci);
		\draw[->] (ci) -- (ci');
		\draw[->, blue,dashed] (ci) -- (0,-7);
		\end{tikzpicture}
		}
		\caption{Symbolic (infinite) state space when computing \textPET{} in a naive way.}
		\label{figure:1-clock_pta:ss}
	\end{subfigure}
	\caption{A 1-clock PTA and the \textPET{} problem.}
	\label{figure:1-clock_pta}
\end{figure}

Before discussing how this regular expression can be used to answer the \FullOpacityParamEmptinessProblem, let us illustrate how it is obtained on a simple example.
\cref{figure:1-clock_pta:pta} depicts a 1-clock PTA~$\PTA$ with a clock $\clock$ and two parameters~$p$ and~$q$.
We are interested in solving $\PET{\PTA}$ where we assume here that $\locfinal$ is~$\loc_1$.
Applying the semi-algorithm from \cref{ss:PETS}, suppose the addition of a clock $\clockabs$ and parameter $\paramabs$ to the PTA, followed by the computation of the reachability synthesis to~$\loc_1$.
In this case, the algorithm does not terminate though, and as shown in \cref{figure:1-clock_pta:ss}.

Following the steps of our method, we have $\Rfp(\PTA, \loc_1) = \{(\loc_0, \loc_0),(\loc_0, \loc_1)\}$.
\cref{figure:transformed:l_0l_0,figure:transformed:l_0l_1} depict the corresponding reset-free automata while \cref{figure:transformed:zone_automaton} gives the automaton of the zones.
Urgent locations are colored in yellow.

\begin{figure} [tb]
	\centering
	\small
	\begin{subfigure}[b]{0.32\textwidth}
		\centering
		\begin{tikzpicture}[PTA, thin]

			\node[location, initial] at(0,0) (l0') {$\loc_0'$};
			\node[location] at(1.5,0) (l1) {$\loc_1$};
			\node[location,urgent] at(0,1.5) (l0) {$\loc_0$};
			\node [invariant, below] at (0,-0.5) {$x \leq p$};

			\path
				(l0') edge[] node [above] {$x \geq q$} (l1)
				(l0') edge[] node[left,align=center] {$x = p$\\$x = \paramabs$}  (l0)
			;
		\end{tikzpicture}
		\caption{$\PTA(\loc_0,\loc_0)$}
		\label{figure:transformed:l_0l_0}
	\end{subfigure}
	\begin{subfigure}[b]{0.32\textwidth}
		\centering
		\begin{tikzpicture}[PTA, thin]

			\node[location, initial] at(0,0) (l0) {$\loc_0$};
			\node[location,urgent] at(1.5,0) (l1') {$\loc_1'$};
			\node[location,urgent] at(1.5,1.5) (l1) {$\loc_1$};
			\node [invariant, below] at (0,-0.5) {$x \leq p$};

			\path
				(l0) edge[] node {$x \geq q$}  (l1')
				(l1') edge[] node {$x = \paramabs$} (l1)
			;
		\end{tikzpicture}
		\caption{$\PTA(\loc_0,\loc_1)$}
		\label{figure:transformed:l_0l_1}
	\end{subfigure}
	\begin{subfigure}[b]{0.32\textwidth}
		\centering
		\begin{tikzpicture}[PTA, thin]

			\node[location, initial] at(0,0) (l0) {$\loc_0$};
			\node[location] at(2,0) (l1) {$\loc_1$};

			\path
				(l0) edge[loop] node [above] {$Z_{\loc_0,\loc_0}$} (l0)
				(l0) edge[] node {$Z_{\loc_0,\loc_1}$}  (l1)
			;
		\end{tikzpicture}
		\caption{Automaton of the zones~$\hat{\PTA}$}
		\label{figure:transformed:zone_automaton}
	\end{subfigure}
	\caption{Reset-free automata of~$\PTA$ (from \cref{figure:1-clock_pta:pta}) and automaton of the zones~$\hat{\PTA}$.}
	\label{figure:transformed}
\end{figure}
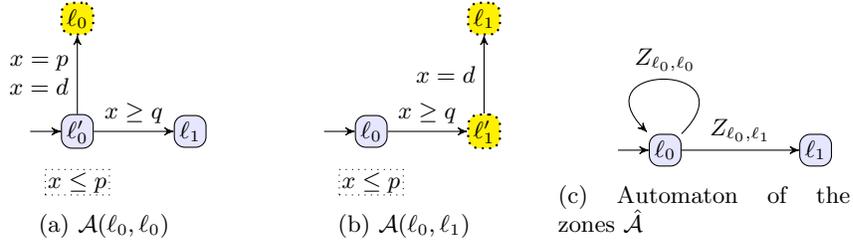

Reachability synthesis of the reset-free automata gives $Z_{\loc_0,\loc_0} = \{\paramabs = p\}$ and $Z_{\loc_0,\loc_1} = \{q \leq \paramabs \leq p\}$.
As per \cref{language equivalence}, the expression $\bar{e}$ (obtained by replacing operators in the regular expression of the language of $\hat{\PTA}$) is equivalent to $\PET{\PTA}$ (again taking $\loc_1$ as final location). %
That expression can be easily obtained (for example with a state elimination method) and gives $\bar{e} = (Z_{\loc_0,\loc_0})^{\bar{*}}\ \bar{.}\ Z_{\loc_0,\loc_1}$.
We may then develop operations on $\bar{e}$ and obtain the following infinite disjunction of parametric zones.

\begin{equation}
\begin{array}{ccl}
\PET{\PTA}	& = &		(Z_{\loc_0,\loc_0})^{\bar{*}}\ \bar{.}\ Z_{\loc_0,\loc_1} \\
				& = &		\{\paramabs = p\}^{\bar{*}} \bar{.}\ \{q \leq \paramabs \leq p\} \\
				& = &		\{\paramabs = 0 \lor \paramabs = p \lor \paramabs = 2p \lor \dots\} . \{q \leq \paramabs \leq p\} \\
				& = &		\{q \leq \paramabs \leq p \lor q + p \leq \paramabs \leq 2p \lor q + 2p \leq \paramabs \leq 3p \lor \dots\}

\end{array}
\end{equation}

\subsection{Solving the $\problemFTOE$ problem through a translation of \textPET{} to parametric Presburger arithmetic}\label{ss:PresburgerFTOE}

Presburger arithmetic is the first order theory of the integers with addition.
It is a useful tool that can represent and manipulate sets of integers called semi-linear sets.
Those sets are particularly meaningful to study TAs, as the set of durations of runs reaching the final location can be described by a semi-linear set~\cite{BDR08}.
Presburger arithmetic is however not expressive enough to represent durations of runs in PTAs due to the presence of parameters.
In~\cite{L23}, a parametric extension of Presburger arithmetic was considered, introducing linear parametric semi-linear sets (\LpSl{} sets) which are functions associating to a parameter valuation $\pval$ a (traditional) semi-linear set of the following form:
\begin{align}
\label{eq:semip}
\begin{split}
S(\pval) = \Big \{x\in \setN^m\mid \bigvee_{i\in I} \exists x_0,\dots x_{n_i}\in \setN^m,& k_1,\dots k_{n_i}\in \setN,
x=\sum_{j=0}^{n_i} x_j\\ \wedge b_0^i(\pval)\leq x_0\leq &c_0^i(\pval)
 \wedge\bigwedge_{j=1}^{n_i}
k_jb_j^i(\pval)\leq x_j \leq k_j c_j^i(\pval) \Big \}
\end{split}
\end{align}
\noindent{}where $I$ is a finite set and the $b_j^i$ and $c_j^i$ are linear polynomials with coefficients in~$\setN$.
A~1-\LpSl{} set is an \LpSl{} set defined over a single parameter.
Given two \LpSl{} (resp.\ 1-\LpSl{}) sets $S_1$ and~$S_2$, the \LpSl{} (resp.\ 1-\LpSl{}) equality problem consists in deciding whether there exists a parameter valuation $\pval$ such that $S_1(\pval)=S_2(\pval)$.

\begin{theorem}[\cite{L23}]
\label{th:LpSl}
The \LpSl{} equality problem is undecidable.

The 1-\LpSl{} equality problem is decidable.
Moreover, the set of valuations achieving equality can be computed.
\end{theorem}

The main goal of this subsection is to relate the expressions computed in \cref{ss:infinitePETS} to \LpSl{} sets in order to tackle \OpacityTextForSec{} problems.
Since Presburger arithmetic is a theory of integers, we have to restrict PTAs to integer parameters; this is what prevents our results to be extended to rational-valued parameters in a straightforward manner.
Moreover, we need to focus on time durations of runs with integer values.
This second restriction however is without loss of generality.
Indeed, in~\cite[Theorem~5]{ALM23}, a trick is provided
(which consists mainly in doubling every term of the system so that any run duration that used to be a rational of the form $\frac q 2$ is now an integer
to ensure that if a set is non-empty, it contains an integer. This transformation also allows one to consider only non-strict constraints, and thus we assume every constraint is non-strict in the following.
\begin{restatable}{theorem}{FTOELPSL}
\label{th:FTOE-LpSl}
The \LpSl{} equality problem reduces to the $\problemFTOE$ problem for $(1, 0, \arbitrarilyMany)$-PTAs\LongVersion{ with integer-valued parameters over dense time}. %

Moreover, the $\problemFTOE$ problem for $(1, 0, 1)$-PTAs\LongVersion{ with integer-valued parameters over dense time} %
reduces to the 1-\LpSl{} equality problem.
\end{restatable}
\begin{proof}[Sketch of proof]
From \cref{eq:semip} one can see that an 
\LpSl{} set parametrically defines integers that are the sum of two types of elements:
$x_0$ belongs to an interval, while the $x_j$ %
represent a sum of integers, each coming from the interval $[b_j^i;c_j^i]$. 
Intuitively, we separate a run into its elementary path until the final state and its loops.
We use $x_0$ to represent the duration of the elementary path, and the $x_j$ adds the 
duration of loops. Each occurrence of the same loop within a run being independent (as they 
include a reset of the clock), their durations all belong to the same interval.

Formally, given a PTA~$\PTA$, using~\cref{ss:problems}, we build the PTAs 
$\PTA^{\locpriv}_{\locfinal}$ and $\PTA^{\neg\locpriv}_{\locfinal}$
separating the private and public runs of~$\PTA$. 
Then with~\cref{ss:infinitePETS}, we
obtain expressions $\bar{e}_{\locpriv}$ and $\bar{e}_{\neg\locpriv}$
such that (\cref{language equivalence}) 
$\bar{e}_{\locpriv} = \PET{\PTA^{\locpriv}_{\locfinal}}$ and
$\bar{e}_{\neg\locpriv} = \PET{\PTA^{\neg\locpriv}_{\locfinal}}$.
We then develop and simplify these expressions until we can build \LpSl{} sets representing
the integers accepted by each expression. We can then show the inter-reduction as 
the \fullOpacityTextForSec{} is directly equivalent to the equality of the two sets.
Note that one direction of the reduction is stronger, allowing multiple parameters. 
This is due to constraints over the parameters which may appear in our expressions, but
cannot be transferred to \LpSl{} sets. However, when there is a single parameter, one can
easily resolve these constraints beforehand.
\end{proof}

Combining \cref{th:FTOE-LpSl,th:LpSl} directly gives us:

\begin{corollary}\label{corollary-FTOE-undecidable}
$\problemFTOE$ is undecidable for $(1, 0, \arbitrarilyMany)$-PTAs\LongVersion{ with integer-valued parameters over dense time}.
\end{corollary}
\begin{corollary}\label{corollary-FTOE-decidable}
$\problemFTOE$ is decidable for $(1, 0, 1)$-PTAs and $\problemFTOS$ can be solved.
\end{corollary}
\section{Decidability of $\problemTOE$ for $(1, 0, \arbitrarilyMany)$-PTAs for integer-valued parameters}\label{section:TOE}

We prove here the decidability of %
\problemTOE{} for $(1, 0, \arbitrarilyMany)$-PTAs with integer parameters over dense time (\cref{ss:TOE:Presburger}); we also prove that the same problem is \LongVersion{decidable }in \EXPSPACE{} for $(1, \arbitrarilyMany, 1)$-PTAs over discrete time (\cref{ss:TOE:GH21}).

\subsection{General case}\label{ss:TOE:Presburger}

Adding the divisibility predicate (denoted ``$|$'') to Presburger arithmetic produces an undecidable theory, whose purely existential fragment is known to be decidable~\cite{LOW15}.
The $\problemFTOE$ problem can be encoded in this logic, but requires a single quantifier alternation, which goes beyond the aforementioned decidability result, leading us to rely on~\cite{L23}.
The $\problemTOE$ problem however can be encoded in the purely existential fragment.

\begin{restatable}{theorem}{TOEdiv}\label{theorem:TOE-div}
	The $\problemTOE$ problem is decidable.
\end{restatable}

\begin{proof}[Sketch of proof]
As for \cref{th:FTOE-LpSl}, we start by building and simplifying expressions
representing the private and public durations of the PTA.
Instead of translating the expression into \LpSl{} set however, we now use Presburger with 
divisibility.

Again, a run can be decomposed in the run without loops, and its looping parts. 
The duration of the former is defined directly by conjunction of inequalities, 
which can be formulated in a Presburger arithmetic formula. The latter requires the 
divisibility operator to represent the arbitrary number of loops.
Hence, we can build a formula accepting exactly the integers satisfying our expressions.
Deciding the $\problemTOE$ problem can be achieved by testing the existence of an 
integer satisfying the formulas produced from both expressions, which can be stated in a 
purely existential formula.
\end{proof}
\begin{remark}[complexity]\label{remark:TOE-div:complexity}
Let us quickly discuss the complexity of this algorithm.
The expressions produced by \cref{language equivalence} can, in the worst case, be exponential in the size of the \PTAtext{}.
This formula was then simplified within the proof of~\cref{th:FTOE-LpSl}, in part by developing it, which could lead to an exponential blow-up.
Finally, the existential fragment of Presburger arithmetic with divisibility can be solved
in \NEXPTIME~\cite{LOW15}. 
As a consequence, our algorithm lies in \threeNEXPTIME{}.

\end{remark}
\subsection{Discrete time case}\label{ss:TOE:GH21}

There are clear ways to improve the complexity of this algorithm.
In particular, we finally prove an alternative version of \cref{theorem:TOE-div} in a more restricted setting ($\Time = \setN$), but with a significantly lower complexity upper bound %
and using completely different proof ingredients~\cite{GH21}.

\begin{restatable}{theorem}{thEOEeasy}
\label{theorem:TOE-div:GH21}
	$\problemTOE$ is decidable in \EXPSPACE{} for $(1, \arbitrarilyMany, 1)$-PTAs over discrete time.
\end{restatable}
\begin{remark}
	The fact that we can handle arbitrarily many non-parametric clocks in \cref{theorem:TOE-div:GH21} does not improve \cref{theorem:TOE-div}:
	over discrete time, it is well-known that non-parametric clocks can be eliminated using a technique from~\cite{AHV93}, and hence come ``for free''.
\end{remark}
\section{Conclusion and perspectives}\label{section:conclusion}

\LongVersion{%
	Recall that we denote by $(pc, npc, p)$-PTAs the class of PTAs where $pc$ is the number of parametric clocks (\ie{} compared at least once in a guard 	to a parameter), $npc$ the number of non-parametric clocks and $p$ the number of parameters (``$\arbitrarilyMany$'' denotes ``any number).
}

\begin{table}[tb]
\caption{Execution-time opacity problems for PTAs: contributions and some open cases}%
\label{table:summary-problems}%
{\centering
\scriptsize
\setlength{\tabcolsep}{1pt} %

\crefname{theorem}{\text{Th.}}{\text{Th.}} %
\crefname{corollary}{\text{Corol.}}{\text{Corol.}} %

\begin{tabular}{|c|c|c|c|c| }
\hline
	\rowHeader{} Time & $(pc, npc, p)$  & \problemTOE{} emptiness  & \problemTOE{} synthesis \\
\hline
dense & $(1, 0, \arbitrarilyMany)$ & \cellYes{} (\cref{theorem:TOE-div}) & \cellOpen{}\\
\hline
dense & $(1, \arbitrarilyMany, \arbitrarilyMany)$ & \cellOpen{} & \cellOpen{}\\
\hline
dense & $(2, 0, 1)$ & \cellOpen{} & \cellOpen{}\\
\hline
dense & $(3, 0, 1)$ & \cellNo{} (\cite[Th.6.1]{ALMS22}) & \cellNo{}\\
\hline
discrete & $(1, \arbitrarilyMany, 1)$ & \cellYes{\EXPSPACE{}} (\cref{theorem:TOE-div:GH21}) & \cellOpen{} \\
\hline
\end{tabular}
\hfill{}
\begin{tabular}{ |c|c|c|c|c| }
\hline
	\rowHeader{} Time & $(pc, npc, p)$  & \problemFTOE{} emptiness  & \problemFTOE{} synthesis \\
\hline
dense & $(1, 0, 1)$ & \cellYes{} \cref{corollary-FTOE-decidable} & \cellYes{} \cref{corollary-FTOE-decidable}\\
\hline
dense & $(1, 0, [2,M))$ & \cellOpen{} & \cellOpen{}\\
\hline
dense & $(1, 0, M)$ & \cellNo{} (\cref{corollary-FTOE-undecidable}) & \cellNo{}\\
\hline
dense & $([2,3], 0, 1)$ & \cellOpen{} & \cellOpen{}\\
\hline
dense & $(4, 0, 2)$ & \cellNo{} (\cite[Th.\,7.1]{ALMS22}) & \cellNo{}\\
\hline
\end{tabular}

}
\end{table}

\LongVersion{%
\subsection{Summary of contributions}
}

In this paper, we addressed the \opacityText{} for 1-clock PTAs with integer-valued parameters over dense time.
We proved that
\begin{oneenumerate}%
	\item \problemFTOE{} is undecidable for a sufficiently large number of parameters,
	\item \problemFTOE{} becomes decidable for a single parameter, and
	\item \problemTOE{} is decidable, in \threeNEXPTIME{} over dense time and in \EXPSPACE{} over discrete time.
\end{oneenumerate}%
These results rely on a novel construction of \textPET{}, for which a sound and complete computation method is provided\LongVersion{ (that also works for rational-valued parameters)}.
In the general case, we provided semi-algorithms for the computation of \textPET{}, \problemTOS{} and \problemFTOS{}.

Our \textPET{} constructions and all \textPET{}-related results work perfectly for rational-valued parameters.
It remains however unclear how to extend our (un)decidability results to rational-valued parameters, as our other proof ingredients (notably using the Presburger arithmetics) heavily rely on integer-valued parameters.

\LongVersion{%
 \subsection{Perspectives}
}

It remains also unclear whether synthesis can be achieved using techniques from~\cite{GH21}, explaining the ``open'' cell in the ``discrete time'' row of~\cref{table:summary-problems}.
Also, a number of problems remain open in \cref{table:summary-problems}, notably the 2-clock case, already notoriously difficult for reachability emptiness~\cite{AHV93,GH21}.

Finally, exploring \emph{weak} ET-opacity~\cite{ALLMS23} (which allows the attacker to deduce that the private location was \emph{not} visited) is also on our agenda.

	\newcommand{\CCIS}{Communications in Computer and Information Science}
	\newcommand{\ENTCS}{Electronic Notes in Theoretical Computer Science}
	\newcommand{\FAC}{Formal Aspects of Computing}
	\newcommand{\FundInf}{Fundamenta Informaticae}
	\newcommand{\FMSD}{Formal Methods in System Design}
	\newcommand{\IJFCS}{International Journal of Foundations of Computer Science}
	\newcommand{\IJSSE}{International Journal of Secure Software Engineering}
	\newcommand{\IPL}{Information Processing Letters}
	\newcommand{\JAIR}{Journal of Artificial Intelligence Research}
	\newcommand{\JLAP}{Journal of Logic and Algebraic Programming}
	\newcommand{\JLAMP}{Journal of Logical and Algebraic Methods in Programming} %
	\newcommand{\JLC}{Journal of Logic and Computation}
	\newcommand{\LMCS}{Logical Methods in Computer Science}
	\newcommand{\LNCS}{Lecture Notes in Computer Science}
	\newcommand{\RESS}{Reliability Engineering \& System Safety}
	\newcommand{\STTT}{International Journal on Software Tools for Technology Transfer}
	\newcommand{\TCS}{Theoretical Computer Science}
	\newcommand{\ToPNoC}{Transactions on Petri Nets and Other Models of Concurrency}
	\newcommand{\TSE}{{IEEE} Transactions on Software Engineering}

\ifdefined\VersionAuthor
	\renewcommand*{\bibfont}{\small}
	\printbibliography[title={References}]
\else
	\bibliography{oneclockopacity}
\fi
\newpage
\appendix

\section{Recalling the correctness of \EFsynth}\label{appendix:prop:EFsynth}
\begin{lemma}[\cite{JLR15}]\label{prop:EFsynth}
	Let $\PTA$ be a PTA, and let $\LocsTarget$ be a subset of the locations of~$\PTA$.
	Assume $\EFsynth(\PTA, \LocsTarget)$ terminates with result~$\K$.
	Then $\pval \models \K$ iff $\LocsTarget$ is reachable in~$\valuate{\PTA}{\pval}$.
\end{lemma}
\section{Proof of results}
\subsection{Proof of \cref{proposition:GeneralPETS}}\label{appendix:proof:proposition:GeneralPETS}

\propGeneralPETS*

\begin{proof}
By having $\locfinal$ being urgent and removing its outgoing edges, we ensure that the runs that reach~$\locfinal$ in~$\PTA'$ are all of the form $(\loc_0, \clockval_0), (d_0, \edge_0), \cdots, (\loc_n, \clockval_n)$ for some $n \in \setN$ such that $\loc_n = \loc'$ and $\forall 0 \leq i \leq n-1, \loc_i \neq \loc'$.
By having a clock $\clockabs$ that is never reset and $\locfinal$ being urgent, we ensure that for any run $\varrun$ that reaches $\locfinal$ in $\PTA'$, the value of $\clockabs$ in the final state if equals to $\duration(\varrun)$.
By having a guard $\clockabs = \paramabs$ on all incoming edges to $\locfinal$, we ensure that $\paramabs = \duration(\varrun)$ on any run $\varrun$ that reaches~$\locfinal$.

Therefore, $\EFsynth(\PTA', \{\locfinal\})$ contains all parameter valuations of the runs to $\locfinal$ in $\PTA$ that stop once $\locfinal$ is reached, along with the duration of those runs contained in $\paramabs$.
\end{proof}
\subsection{Proof of \cref{TOS}}\label{appendix:proof:proposition:TOS}

\propTOS*

\begin{proof}
By definition, $\problemdTOS(\PTA)$ is the synthesis of parameter valuations~$\pval$ and execution times~$\execTimes_\pval$ such that $\valuate{\PTA}{\pval}$ is opaque \wrt{} $\locpriv$ on the way to~$\locfinal$ for these execution times~$\execTimes_\pval$.
This means that $\problemdTOS(\PTA)$ contains exactly all parameter valuations and execution times for which there exist both at least one run in $\PTA^{\locpriv}_{\locfinal}$ and at least one run in $\PTA^{\neg \locpriv}_{\locfinal}$.
Since \textPET{} are the synthesis of the parameter valuations and execution times up to the final location, $\problemdTOS(\PTA)$ is equivalent to the intersection of the $\PET{\PTA^{\locpriv}_{\locfinal}}$ and $\PET{\PTA^{\neg \locpriv}_{\locfinal}}$.
\end{proof}
\subsection{Proof of \cref{fullTOS}}\label{appendix:proof:proposition:fullTOS}

\propFullTOS*

\begin{proof}
By definition, $\problemdFTOS(\PTA)$ is the synthesis of parameter valuations~$\pval$ (and execution times of their runs) \suchthat{} $\valuate{\PTA}{\pval}$ is fully opaque \wrt{} $\locpriv$ on the way to~$\locfinal$.
By definition, $\projectP{\Diff(\PTA)}$ is the set of parameter valuations \suchthat{} for any valuation $\pval \in \projectP{\Diff(\PTA)}$, there is at least one run where $\locpriv$ is reached (resp.\ avoided) on the way to~$\locfinal$ in $\valuate{\PTA}{\pval}$ whose duration time is different from those of any run where $\locpriv$ is avoided (resp.\ reached) on the way to~$\locfinal$ in $\valuate{\PTA}{\pval}$.
By removing this set of parameters from $\problemdTOS(\PTA)$, we are left with parameter valuations (and execution times of their runs) \suchthat{} for any $\pval$, any run $\varrun$ where $\locpriv$ is reached (resp.\ avoided) on the way to~$\locfinal$ in $\valuate{\PTA}{\pval}$, there is a run $\varrun'$ where $\locpriv$ is avoided (resp.\ reached) on the way to~$\locfinal$ in $\valuate{\PTA}{\pval}$ and $\duration(\varrun) = \duration(\varrun')$.
This is equivalent to our definition of full opacity.
\end{proof}
\subsection{\cref{equiv_not_lf}}\label{appendix:proof:equiv_not_lf}
\begin{restatable}{proposition}{propequivnotlf}
\label{equiv_not_lf}
	Let~$\PTA$ be a 1-clock PTA, and $(\loc_i,\loc_j) \in \Rfp(\PTA, \locfinal)$ such that $\loc_j \neq \locfinal$.
	Then $Z_{\loc_i,\loc_j}$ is equivalent to the synthesis of parameter valuations~$\pval$ and execution times~$\execTimes_\pval$ such that $\execTimes_\pval = \{ d \mid \exists \varrun$ from $(\loc_i,\{\clock = 0\})$ to $\loc_j$ in $\valuate{\PTA}{\pval}$ such that  $d = \duration(\varrun)$, $\locfinal$ is never reached, and $\clock$ is reset on the last edge of $\varrun$ and on this edge only $\}$.
\end{restatable}
\begin{proof}
Let us first consider the case where $\loc_i \neq \loc_j$.
Steps 1 to~3 in \cref{sub-automata} imply that whenever $\loc_j$ occurs either as a source or target location in an edge, it is replaced by the duplicate locality $\loc_j'$, except when $\loc_j$ is the target location and $\clock$ is reset on the edge.
At this stage, for any path between $\loc_i$ and $\loc_j$ in $\PTA$, where no incoming edge to $\loc_j$ featuring a clock reset is present, there is an equivalent path in $\PTA(\loc_i,\loc_j)$ with $\loc_j$ being replaced by $\loc_j'$.
Step~4 implies that whenever $\loc_j$ is reached in $\PTA(\loc_i,\loc_j)$ no delay is allowed.
As there are no outgoings edges from $\loc_j$ anymore, and only incoming edges featuring a clock reset, only runs ending with such edges are accepted by the reachability synthesis on $\loc_j$.
Since the clock value when entering in $\loc_j$ through such an edge is always $0$, removing the upper bound of the invariant does not impact the availability of transitions.
Because of our assumption that $\loc_i \neq \loc_j$, Step 5 does not change the initial location.
Step~6 ensures that, in any run from $\loc_i$ to $\loc_j$ :
\begin{itemize}
\item no clock reset is performed before the last edge of the run;
\item the clock is not reset when entering $\loc_j$, and is therefore equals to the duration of the run;
\item $\locfinal$ is not reached.
\end{itemize}
Step~7 ensures that $\paramabs$ is equal to the value of the clock when entering $\locfinal$.

\smallskip

Let us now consider the case where $\loc_i = \loc_j$.
In this case, Step~5 changes the initial locality to $\loc_j'$.
Because of Steps 1 to~3, runs from $\loc_j'$ to $\loc_j$ in $\PTA(\loc_i,\loc_j)$ are identical to runs looping from $\loc_i$ to $\loc_i$ in $\PTA$ where $\clock$ is reset on the last edge of the run and on this edge only.
Restrictions obtained by Steps 4, 6 and~7 are unchanged.

\smallskip

Therefore, $Z_{\loc_i,\loc_j}$ is equivalent to the synthesis of parameter valuations~$\pval$ and execution times~$\execTimes_\pval$ such that $\execTimes_\pval = \{ d \mid \exists \varrun$ from $(\loc_i,\{\clock = 0\})$ to $\loc_j$ in $\valuate{\PTA}{\pval}$ such that  $d = \duration(\varrun)$, $\locfinal$ is never reached, and $\clock$ is reset on the last edge of $\varrun$ and on this edge only.
\end{proof}
\subsection{\cref{equiv_lf}}\label{appendix:proof:equiv_lf}
\begin{restatable}{proposition}{propequivlf}
\label{equiv_lf}
	Let~$\PTA$ be a 1-clock PTA, and $(\loc_i,\loc_j) \in \Rfp(\PTA, \locfinal)$ such that $\loc_j = \locfinal$.
	Then $Z_{\loc_i,\loc_j}$ is equivalent to the synthesis of parameter valuations~$\pval$ and execution times~$\execTimes_\pval$ such that $\execTimes_\pval = \{ d \mid \exists \varrun$ from $(\loc_i,\{\clock = 0\})$ to $\locfinal$ in $\valuate{\PTA}{\pval}$ such that  $d = \duration(\varrun)$, $\locfinal$ is reached only on the last state of $\varrun$, and $\clock$ may only be reset on the last edge of $\varrun$ $\}$.
\end{restatable}
\begin{proof}
By \cref{def_rfp}, we know that $\loc_i \neq \locfinal$.

Steps 1 to~3 in \cref{sub-automata} imply that:
\begin{itemize}
\item whenever $\locfinal$ is the target location of an edge, it is replaced by the duplicate locality $\loc_j'$, except when $\clock$ is reset on the edge;
\item once $\loc_j'$ is reached, no delay is allowed and the only available transition consists in reaching $\locfinal$ through an empty action $\epsilon$.
\end{itemize}
At this stage, the only difference between path from $\loc_i$ to $\locfinal$ in $\PTA(\loc_i,\loc_j)$ and $\PTA$ is that incoming edges to $\locfinal$ where $\clock$ is not reset now leads to $\loc_j'$, and then to $\locfinal$ without any added elapsed time.
Step~4 implies that whenever $\loc_f$ is reached in $\PTA(\loc_i,\loc_j)$ no delay is allowed.
As $\locfinal$ is either entered by the immediate transition from $\loc_j'$ or feature a clock reset, removing the upper bound of the invariant does not impact the availability of transitions.
As $\loc_i \neq \locfinal$, Step~5 does not change the initial location.
Step~6 ensures that, in any run from $\loc_i$ to $\loc_j$ :
\begin{itemize}
\item no clock reset is performed before the last edge of the run (not counting the $\epsilon$ edge from $\loc_j'$ to $\locfinal$);
\item the clock value is not reset when entering $\locfinal$, and is therefore equal to the duration of the run;
\item no action can be taken after reaching $\locfinal$.
\end{itemize}
Step~7 ensures that $\paramabs$ is equal to the value of the clock when entering $\locfinal$.

\smallskip

Therefore, $Z_{\loc_i,\loc_j}$ is equivalent to the synthesis of parameter valuations~$\pval$ and execution times~$\execTimes_\pval$ such that $\execTimes_\pval = \{ d \mid \exists \varrun$ from $(\loc_i,\{\clock = 0\})$ to $\locfinal$ in $\valuate{\PTA}{\pval}$ such that  $d = \duration(\varrun)$, $\locfinal$ is reached only on the last state of $\varrun$, and $\clock$ may only be reset on the last edge of $\varrun$.

\end{proof}
\subsection{Proof of \cref{language equivalence}}\label{appendix:proof:language equivalence}

\proplanguageequivalence*

\begin{proof}
Let us first show that $\bar{e}$ contains $\PET{\PTA}$.
Let $\rho$ be a path whose time duration and parameter constraints are in $\PET{\PTA}$.
By definition, $\rho$ starts at time 0 in the initial locality and ends in $\locfinal$, with only one occurrence of $\locfinal$ in the whole path.
Let us consider that the clock is reset $n$ times before the last transition, then $\rho$ can be decomposed as $\rho_0\dots\rho_n$ such that:
\begin{itemize}
\item $\forall\ 0 \leq i < n$, sub-path $\rho_i$ starts in $\loc_i$ at time valuation 0, ends in $\loc_{i+1}$, contains a single reset positioned on the last transition (thus ending with time valuation 0) and does not contain any occurrence of $\locfinal$;
\item sub-path $\rho_n$ starts in $\loc_n$ at time valuation 0, ends in $\locfinal$, may only contain a reset on its last transition, and contains exactly one occurrence of $\locfinal$.
\end{itemize}
By \cref{def_rfp}, $\forall\ 0 \leq i < n$, $(\loc_i,\loc_{i+1}) \in \Rfp(\PTA, \locfinal)$ and by \cref{equiv_not_lf}, $Z_{\loc_i,\loc_{i+1}}$ is the synthesis of parameter valuations and execution times of that sub-path.
By \cref{def_rfp}, $(\loc_n,\locfinal) \in \Rfp(\PTA, \locfinal)$ and by \cref{equiv_lf}, $Z_{\loc_n,\locfinal}$ is the synthesis of parameter and valuation times of that sub-path.
By \cref{zone_automaton}, there is a sequence of transitions $Z_{\loc_0,\loc_1},\dots, Z_{\loc_i,\loc_{i+1}}, \dots, Z_{\loc_n,\locfinal}$ in the automaton of the zones $\hat{\PTA}$.
By application of operators $\bar{+}$ and $^{\bar{*}}$, that sequence thus exists in $\bar{e}$ as $Z_{\loc_0,\loc_1}\bar{.}\dots\bar{.} Z_{\loc_i,\loc_{i+1}}\bar{.}\dots\bar{.}Z_{\loc_n,\locfinal}$.
By definition of operator $\bar{.}$, this expression is the intersection of all parameter constraints and the addition of all valuation times, which is equivalent to $\PET{\PTA}$.

\smallskip

Let us now show that $\PET{\PTA}$ contains $\bar{e}$.
By application of operators $\bar{+}$ and $^{\bar{*}}$, any word in $\bar{e}$ can be expressed as a sequence of concatenation operations $\bar{.}$.
By \cref{zone_automaton}, given a word $Z_{\loc_0,\loc_1}\bar{.}\dots\bar{.} Z_{\loc_i,\loc_{i+1}}\bar{.}\dots\bar{.}Z_{\loc_n,\loc{n+1}} \in \bar{e}$, we know that $\loc_0$ is the initial location of~$\PTA$, $\loc_{n+1} = \locfinal$ and $\forall\ 0 \leq i \leq n, \loc_{i} \neq \locfinal$.
By \cref{equiv_not_lf}, $\forall\ 0 \leq i < n$, $Z_{\loc_i,\loc_{i+1}}$ is the synthesis of parameter valuations and execution times of paths between $\loc_i$ and $\loc_{i+1}$  in $\PTA$ such that $\locfinal$ is never reached, and $\clock$ is reset on the last edge of the path and on this edge only.
And by \cref{equiv_lf}, $Z_{\loc_n,\locfinal}$ is the synthesis of parameter valuations and execution times of paths between $\loc_n$ and $\locfinal$ in $\PTA$ such that $\locfinal$ is reached only on the last state of $\varrun$, and $\clock$ may only be reset on the last edge of $\varrun$.

Let us assume there exists a path $\rho$ whose time duration and parameter constraints are in $\PET{\PTA}$ such that $\rho = \rho_0\dots\rho_n$ and:
\begin{itemize}
\item $\forall\ 0 \leq i < n$, sub-path $\rho_i$ starts in $\loc_i$ at time valuation 0, ends in $\loc_{i+1}$, contains a single reset positioned on the last transition (thus ending with time valuation 0) and does not contain any occurrence of $\locfinal$;
\item sub-path $\rho_n$ starts in $\loc_n$ at time valuation 0, ends in $\locfinal$, may only contain a reset on its last transition, and contains exactly one occurrence of $\locfinal$.
\end{itemize}
Then $Z_{\loc_0,\loc_1}\bar{.}\dots\bar{.} Z_{\loc_i,\loc_{i+1}}\bar{.}\dots\bar{.}Z_{\loc_n,\loc{n+1}} \in \PET{\PTA}$.
On the other hand, if there does not exist such a path, then there exist $ 0 \leq i \leq n$ such that $Z_{\loc_{i},\loc_{i+1}} = \emptyset$.
By recursive applications of operator $\bar{.}$, the whole sequence is evaluated as $\emptyset$ and thus contained in $\PET{\PTA}$.

\end{proof}
\subsection{Proof of \cref{th:FTOE-LpSl}}

\FTOELPSL*
\begin{proof}
Given a PTA~$\PTA$, we showed in \cref{ss:problems} how to compute two PTAs
$\PTA^{\locpriv}_{\locfinal}$ and $\PTA^{\neg\locpriv}_{\locfinal}$
separating the private and public runs of~$\PTA$. Then in \cref{ss:infinitePETS}, we
showed how to build expressions $\bar{e}_{\locpriv}$ and $\bar{e}_{\neg\locpriv}$
such that (\cref{language equivalence}) 
$\bar{e}_{\locpriv} = \PET{\PTA^{\locpriv}_{\locfinal}}$ and
$\bar{e}_{\neg\locpriv} = \PET{\PTA^{\neg\locpriv}_{\locfinal}}$.

Note that the operators $\bar{.}$, $\bar{*}$ and $\bar{+}$ are associative and  commutative; moreover, each term $Z$ occurring in the expressions $\bar{e}_{\locpriv}$ and $\bar{e}_{\neg\locpriv}$ is a union of constraints
$Z = \bigcup_{i'} {\Constr_{i'}} = \text{\Huge $\bar{+}$}_{i'} \Constr_{i'}$.
As a consequence, we can thus develop the entire expression to the form 
\[
{\text{\Huge $\bar{+}$}}_i\ (\Constr^i_1\bar{.}\Constr^i_2\bar{.}\cdots \bar{.}\Constr^i_{n_i})\bar{.}(\Constr^i_{n_i+1})^{\bar{*}}\bar{.}(\Constr^i_{n_i+2})^{\bar{*}}\bar{.}\cdots\bar{.}(\Constr^i_{n_i+m_i})^{\bar{*}}.
\]
where we put all $\bar{+}$ outside of the expression.
For example, the expression $Z_1\bar{.}(Z_2)^{\bar{*}}$ where $Z_1 = \Constr_1\cup\Constr_2$
and $Z_2 = \Constr_3\cup\Constr_4$ is developed into
$
\Constr_1\bar{.}(\Constr_3)^{\bar{*}}\bar{.}(\Constr_4)^{\bar{*}} \bar{+} 
\Constr_2\bar{.}(\Constr_3)^{\bar{*}}\bar{.}(\Constr_4)^{\bar{*}}.
$

As $\Constr^{\bar{*}} = \{\paramabs=0\}\bar{+}\Constr\bar{.}\Constr^{\bar{*}}$, for each $\Constr^i_{n_i+j}$ we can w.l.o.g.\ express term $i$ as the union of two terms: one where $(\Constr^i_{n_i+j})^{\bar{*}}$ is removed (\ie this loop is never taken), and one where $\Constr^i_{n_i+j}$ is concatenated to the term (\ie the loop is taken at least once).
This means that each term, is turned into $2^{m_i}$ terms, where we can assume w.l.o.g.\ that for each $j>0$, $\Constr^i_{n_i+j} = \Constr^i_j$.

Given an expression of the above form, by definition of $\bar{.}$, the product 
$\Constr^i_1\bar{.}\Constr^i_2\bar{.}\cdots \bar{.}\Constr^i_{n_i}$ is also a conjunction of inequalities and thus can be expressed as
$\Constr_i^\paramabs \cap \Constr_i^\mathbb{P}$ 
where $\Constr_i^\mathbb{P}$ is obtained by the constraints that do not involve $\paramabs$ while $\Constr_i^\paramabs$ contains the constraints that involve $\paramabs$ and potentially some parameters in $\mathbb{P}$.
Note also that by the assumption that for each $j>0$, $\Constr^i_{n_i+j} = \Constr^i_j$,
any constraint that does not involve~$\paramabs$ can be removed from $\Constr^i_{n_i+j}$ without
modifying the set.
Therefore, the expression can now be rewritten as 
\[
{\text{\Huge $\bar{+}$}}_i(\Constr_i^\paramabs \cap \Constr_i^\mathbb{P})\bar{.}(\Constr^i_{1})^{\bar{*}}\bar{.}(\Constr^i_{2})^{\bar{*}}\bar{.}\cdots\bar{.}(\Constr^i_{m_i})^{\bar{*}}.
\]
where every inequality in $\Constr^i_{j}$ involves~$\paramabs$.

\begin{itemize}
	\item
Assume the expressions involve a single parameter $p$.
Let us show that the $\problemFTOE$ problem for PTAs over a single parameter reduces to the 1-\LpSl{} equality problem.

Every constraint on~$p$ is of the form $p\bowtie c$ with $c \in \setN$
and ${\bowtie} \in \{\leq,\geq\}$. Therefore, there exists a constant $M$ such that
for all $i$, either the constraint $\Constr_i^\mathbb{P}$ is satisfied for all $p\geq M$, or
it is satisfied by none.

For any fixed valuation $\pval$, \fullOpacityTextForSec{} of $\valuate{\PTA}{\pval}$ is decidable by~\cite{ALLMS23}. We thus assume that
we consider only valuations of $p$ greater than $M$. 
This can be represented by replacing every occurrence of $p$ in the expressions by $M+p$.
This can be done without loss of generality as we can independently test whether
the PTA is \fullOpaqueText{} for the finitely many integer values of $p$ smaller than $M$.
When solving the $\problemFTOS$ problem, we thus need to include the valuations of $p$ smaller than $M$ that achieved equality to the valuations provided by the reduction.

The terms $\Constr_i^\mathbb{P}$ being either always or never valid, one can either remove
this constraint from the expression, or the term containing it producing an expression of the form 
\[
{\text{\Huge $\bar{+}$}}_i\Constr^i_0\bar{.}(\Constr^i_{1})^{\bar{*}}\bar{.}(\Constr^i_{2})^{\bar{*}}\bar{.}\cdots\bar{.}(\Constr^i_{m_i})^{\bar{*}}.
\]
where every constraint involves~$\clock$.

Once again, assuming $p$ is large enough, the constraint $\Constr_j^i$ can
be assumed to be of the form $\alpha^i_j p + \beta^i_j \leq x \leq \gamma^i_j p+ \delta_j^i$ where $\alpha^i_j, \beta^i_j, \gamma^i_j, \delta_j^i \in \setN$.

For both expressions $\bar{e}_{\locpriv}$ and $\bar{e}_{\neg\locpriv}$, now in the simplified form described above, we build the 1-\LpSl{} sets $S_{\bar{e}_{\locpriv}}$ and
$S_{\bar{e}_{\neg\locpriv}}$
where, taking the notations from \cref{eq:semip}, 
$I$ is the set {\Huge $\bar{+}$} ranges over, for $0\leq j\leq m_i, b_j^i =  \alpha^i_j p + \beta^i_j$ and $c_j^i =  \gamma^i_j p+ \delta_j^i$.

For a valuation $\pval$ of $p$, we have that 
$S_{\bar{e}_{\locpriv}}(\pval)$ contains exactly the integers that satisfy
$\pval(\bar{e}_{\locpriv})$ (and similarly for $S_{\bar{e}_{\neg \locpriv}}(\pval)$ and 
$\pval(\bar{e}_{\neg\locpriv})$).
Therefore, there exists a valuation such that $\PTA$ if fully opaque \wrt{} $\locpriv$ on the way to~$\locfinal$ iff there exists a parameter valuation $\pval$ such that $S_{\bar{e}_{\locpriv}}(\pval)=S_{\bar{e}_{\neg\locpriv}}(\pval)$, establishing the reduction.

\item We now wish to show that
the \LpSl{} equality problem reduces to the $\problemFTOE$ problem.

To do so, we fix two \LpSl{} sets $S_1$ and $S_2$, then build two automata $\PTA_1$ and $\PTA_2$ such that $S_i(\pval)$ contains exactly the integers that satisfy $\pval(\PET{\PTA_i})$, for all valuation $\pval$, for $i\in \{1,2\}$.

Let us focus on $S_1$ and assume it is of the form given by~\cref{eq:semip}.
We build $\PTA_1$ so that from the initial location $\locinit$ it can take
multiple transitions (one for each $i\in I$),
the $i$th transition being allowed if the clock lies between $b^i_0$ and $c^i_0$, reset the clock and reach a state~$\loc_i$.
From $\loc_i$, there are $n_i$ loops, and the $j$th loop can be taken if the clock lies between $b^i_j$ and $c^i_j$ and resets the clock.
Moreover, a transition can be taken from $\loc_i$ to~$\locfinal$ if $x=0$.

Formally, $\PTA_1 = (\ActionSet, \LocSet, \locinit, \ClockSet, \ParamSet, \invariant, \EdgeSet)$
where $\ActionSet=\{\epsilon\}$, $\LocSet = \{\locinit,\locfinal\}\cup \{\loc_i\mid i\in I\}$, $\ClockSet=\{x\}$, $\ParamSet$ is the set of parameters appearing in $S_1$,
$\invariant$ does not restrict the PTA (\ie{} it associates $\setRgeqzero$ to every location), and finally 
\begin{align*}
\EdgeSet=& \big\{(\locinit,(b_0^i\leq x\leq c_0^i) ,\epsilon, \{x\},\loc_i\mid i\in I\big\}\\
\cup & \big\{(\loc_i,(b_j^i\leq x\leq c_j^i) ,\epsilon, \{x\},\loc_i\mid i\in I, 1\leq j\leq n_i\big\}\\
\cup &  \big\{(\loc_i,(x=0) ,\epsilon, \emptyset,\locfinal\mid i\in I\big\}.\\
\end{align*}

Thus, a run reaching $\locfinal$ can be decomposed into final-reset paths. %
In other words, there is a run reaching $\locfinal$ with duration $d$ iff
$d$ can be written as a sum $d=\sum_{j=0}^{n_i} d_j$ where 
$b_0^i\leq d_0\leq c_0^i$ and for all $j>0$, $k_jb_j^i\leq d_j\leq k_jc_j^i$ where $k_j$ is the number of times the $j$th loop is taken in the PTA.
As a consequence, the set of durations of runs reaching $\locfinal$ is exactly~$S_1$.

We build $\PTA_2$ similarly. We now build the PTA~$\PTA$ which can either immediately (with $x=0$) go to the initial state of $\PTA_1$ or go immediately to a private location
$\locpriv$ before immediately reaching the initial state of $\PTA_2$.
The final location of $\PTA_1$ and $\PTA_2$ are then fused in a single location $\locfinal$.
We thus have that, the set of 
runs reaching $\locpriv$ on the way to $\locfinal$ are exactly the ones reaching
$\locfinal$ in $\PTA_2$ (with a prefix of duration~$0$).
And similarly, the set of runs avoiding $\locpriv$ on the way to $\locfinal$ are exactly the ones reaching $\locfinal$ in $\PTA_1$ (with a prefix of duration~$0$).
Therefore, for any parameter valuation $\pval$, we have that 
$\PrivDurVisit{\valuate{\PTA}{\pval}} = \PubDurVisit{\valuate{\PTA}{\pval}} $ iff
$S_1(\pval) = S_2 (\pval)$, concluding the reduction.
\end{itemize}%
\end{proof}
\subsection{Proof of \cref{theorem:TOE-div}}

\newcommand{\TOE}{\ensuremath{\mathit{TOE}}}

\TOEdiv*
\begin{proof}
Within the proof of~\cref{th:FTOE-LpSl}, we considered two expressions
$\bar{e}_{\locpriv}$ and $\bar{e}_{\neg\locpriv}$ such that (\cref{language equivalence}) 
$\bar{e}_{\locpriv} = \PET{\PTA^{\locpriv}_{\locfinal}}$ and
$\bar{e}_{\neg\locpriv} = \PET{\PTA^{\neg\locpriv}_{\locfinal}}$.
Those two expressions were simplified into terms of the form
\[
{\text{\Huge $\bar{+}$}}_i(\Constr_i^\paramabs \cap \Constr_i^\mathbb{P})\bar{.}(\Constr^i_{1})^{\bar{*}}\bar{.}(\Constr^i_{2})^{\bar{*}}\bar{.}\cdots\bar{.}(\Constr^i_{m_i})^{\bar{*}}.
\]
where every inequality in $\Constr^i_{j}$ involves~$\paramabs$.

Assume $\bar{e}_{\locpriv}$ is of the above form, and that for all $i,j$ with $j\leq m_i$,
$\Constr_i^\mathbb{P}=\bigwedge_k I_{i,-1,k}$,
$\Constr_i^d = \bigwedge_k I_{i,0,k}$,
$\Constr^i_{j}=\bigwedge_k I_{i,j,k}$ where each $I_{i,r,k}$ is a linear inequality over $\mathbb{P}$ and $\paramabs$.

We build the formula with free variables $\paramabs, p_1,\dots,p_M,$
\begin{align*}
\phi_{\locpriv} =&  \bigvee_i \exists x_0,\dots x_{m_i}, \paramabs=\sum_{k=1}^{m_i} x_i\\
& \wedge\bigwedge_k  I_{i,-1,k}(p_1,\dots,p_M)\\
& \wedge\bigwedge_k I_{i,0,k}(x_0,p_1,\dots,p_M)\\
& \wedge\bigwedge_j \exists y_1,y_2,y_3,z_1,z_2 (\bigwedge_{m\in\{1,2,3\}}\bigwedge_k I_{i,j,k}(y_m,p_1,\dots,p_M))\\
& \wedge (z_1 = 0 \vee y_1 \ |\ z_1) \wedge (z_2=0\vee y_2 \ |\ z_2) \wedge x_j = z_1+z_2+y_3.
\end{align*}

For fixed values of the variables $p_1,\dots, p_M$, the set of variables $\clock$ satisfying $\phi_{\locpriv}$ is exactly the set of integers contained in $\bar{e}_{\locpriv}$ for parameter valuations $p_1,\dots, p_m$.

Indeed, let us fix one value of~$i$; by definition, the conjunction of constraint $\bigwedge_k  I_{i,-1,k}(p_1,\dots,p_M)$ constrains the variables $p_1,\dots,p_M$ as $\Constr_i^\mathbb{P}$ does to the parameter valuations.
Moreover, by definition of $\bar{.}$, the concatenation of the other constraints accepts the values that can be obtained as a sum of elements produced by each constraint.
This is the role played by the variables $x_i$ in the formulas.

The main point to show is that for $j\geq 1$, the variable $x_j$ takes exactly the values accepted by~$(\Constr^i_{j})^{\bar{*}}$.
Remember that $(\Constr^i_{j})^{\bar{*}}$ accepts every number obtained as a sum of terms accepted by~$\Constr^i_{j}$.

First, by definition, $y_1, y_2$ and $y_3$ all satisfy $\Constr^i_{j}$.
Thus, $z_1$ and $z_2$, being integer multiple of $y_1$ and $y_2$, satisfy $(\Constr^i_{j})^{\bar{*}}$.
Hence, any possible value of $x_j$ belongs to $(\Constr^i_{j})^{\bar{*}}$.

Reciprocally, let $n\in \setN$ accepted by $(\Constr^i_{j})^{\bar{*}}$.
There thus exist $n_1,\dots,n_k$ such that for all~$r$, $n_r$~satisfies $\Constr^i_{j}$ and $n=\sum_{r=1}^k n_r$.
Assume $n_1\leq n_2\leq \dots \leq n_r$.
By convexity of the set described by $\Constr^i_{j}$, every integer between 
$n_1$ and $n_r$ satisfies the constraint.
Thus, we can assume w.l.o.g.\ that at most one number $n_s$ has a value strictly between
$n_1$ and $n_r$ (if two such numbers $a$ and $b$ exist, one can replace them by $a+1$ and $b-1$ to bring them closer to $n_1$ and $n_r$, and by repeating this process, at most one remains).
There thus exist $v_1, v_r\in \setN$ and $v\in [n_1;n_r]$ such that
$n = v_1 n_1 + v_r n_r + v $.
By setting $y_1 = n_1$, $y_2=n_r$, $z_1 = v_1 n_1$, $z_2 =v_r n_r$ and $y_3 = v$,
the variable $x_j$ takes the value $n$.\footnote{The formula allows for $z_1 =0$ and $z_2 =0$, so that if $n$ satisfies $\Constr^i_{j}$, we can set $z_1=z_2=0$ and $y_3 = n$.}

We build $\phi_{\locpub}$ from $\bar{e}_{\locpub}$ in the same way.
Asking whether there exist parameter valuations $p_1,\dots, p_M$ such that an integer $\paramabs \in \setN$ appears in both $\bar{e}_{\locpub}$ and $\bar{e}_{\locpub}$ is thus equivalent to verifying the truth of the formula
\[
\exists p_1,\dots,p_M, \paramabs, \phi_{\locpub}(\paramabs,p_1,\dots,p_M) \wedge \phi_{\locpriv}(\paramabs,p_1,\dots,p_M).
\]
As this formula belongs to the existential fragment of Presburger arithmetic with divisibility, its veracity is decidable, and thus $\problemTOE$ is decidable.

\end{proof}

\subsection{Proof of \cref{theorem:TOE-div:GH21}}

\thEOEeasy*

\begin{proof}
	In~\cite[Section~8]{ALMS22}, we gave a semi-algorithm to answer the $\problemTOS$ problem in $(1, *, 1)$-PTAs, working as follows.
	We build the parallel composition of two occurrences of the input PTA and, adding an absolute time clock, we force simultaneous reachability of the final location such that one PTA visited $\locpriv$ while the other did not.
	This can be reused here, by replacing the absolute time clock with a synchronized action between both PTAs (knowing the actual execution time is not necessary here, as we aim at solving $\problemTOE$---not $\problemTOS$).
	Assuming $\PTA$ is a $(1, *, 1)$-PTA, let~$\PTA'$ denote this resulting PTA.
	Now, from our construction, $\problemTOE$ holds iff the final location of~$\PTA'$ is reachable for at least one parameter valuation.

	Note that, while the (unique) parametric clock of the PTA must be duplicated in~$\PTA'$, the (unique) parameter is not duplicated, as it is the same in both versions of the PTA, and therefore $\PTA'$ contains a single parameter.
	That is, $\PTA'$ is a $(2, *, 1)$-PTA.

	Finally, reachability emptiness is \EXPSPACE{}-complete in $(2, *, 1)$-PTA over discrete time~\cite{GH21}, and therefore the $\problemTOE$ problem for $(1, *, 1)$-PTAs over discrete time can be solved in \EXPSPACE{}.
\end{proof}

\end{document}